\documentclass[a4paper,reqno,onecolumn,11pt,accepted=2026-06-18]{quantumarticle}
\pdfoutput=1

\usepackage[utf8]{inputenc}
\usepackage[T1]{fontenc}
\usepackage{amsmath,bbm,physics,hyperref}
\hypersetup{citecolor=blue}
\usepackage{cleveref}
\usepackage[style=numeric,maxbibnames=99,giveninits=true]{biblatex}
\DefineBibliographyExtras{UKenglish}{}
\usepackage{import}
\usepackage{packages}
\usepackage{definitions}
\abbrevcrefs
\usepackage{silence}
\title[Multi-product Zeno effect]{Multi-product Zeno effect achieving higher order convergence rates}
\author{Tim M\"obus}
\email{moebustim@gmail.com}
\affiliation{Department of Mathematics, University of T\"ubingen, T\"ubingen, Germany}
\affiliation{Department of Mathematics, Technical University of Munich, Munich, Germany}

\begin{document}
	\begin{abstract}
		The quantum Zeno effect is a fundamental mechanism for implementing the effective dynamics of projected Hamiltonian and Lindbladian systems. It approximates the target projected evolution by interleaving Hamiltonian or Lindblad dynamics with quantum operations associated with the desired subspace. In contrast to the related Trotter product formula, the best-known convergence rate of the quantum Zeno effect is typically limited to order $1/n$.
        
        In this work, we improve this convergence rate by employing a multi-product formula, thereby achieving arbitrarily high-order convergence of the form $1/n^{K+1}$. This yields an improved approximation scheme for Zeno-like expectation values via an efficient post-processing method. The approach combines a modified Chernoff lemma, an adapted Dunford-Segal approximation, holomorphic functional calculus, and Chebyshev interpolation. We illustrate the method with the bosonic cat code and also consider the broader class of systems governed by the Bang-Bang decoupling method.
	\end{abstract}
\maketitle

\vspace*{-2ex}

\tableofcontents

\vspace*{\fill}

\section{Introduction}\label{sec:intro}
	The quantum Zeno effect is a fundamental quantum mechanical phenomenon that describes a time evolution frequently interrupted by quantum operations. In its early days, the Zeno effect was understood to freeze a given time evolution, or, in the words of \citeauthor{Misra.1977}, who named the effect after it was first described by \citeauthor{Beskow.1967} \cite{Beskow.1967}:
	\begin{center}
		\textit{An unstable particle which is continuously observed to see whether it decays\\ will never be found to decay!} \cite[Abst.]{Misra.1977}
	\end{center}
	This effect was experimentally verified by the seminal results \cite{Itano.1990,Fischer.2001,Signoles.2014,Bretheau.2015} and many others. In the last two works, it is interesting to note that the Zeno limit is not just a frozen, determined state but rather an ongoing, projected dynamic --- the so-called Zeno dynamic. This convergence behavior has been explored in various studies \cite{Facchi.2004,Facchi.2008,Schmidt.2004}. Even an ``exponential rise of dynamical complexity'' induced by the Zeno effect has been observed \cite{Burgarth.2014}. To highlight this, we express the Zeno effect as the following general limit \cite{Becker.2021,Moebus.2022}:
	\begin{equation}\label{eq:Zeno}
		\Bigl\|\bigl(Me^{\frac{t}{n}\cL}\bigr)^n - \sum_{j=1}^{J}\lambda_j^n e^{t P_j \cL P_j} P_j\Bigr\|_{1 \rightarrow 1} = \cO\Bigl(\frac{1}{n}\Bigr) \qquad \text{for} \qquad n \rightarrow \infty\,,
	\end{equation}
	where $M$ can be any finite quantum operation with $J$ eigenvalues $\lambda_j$ on the unit circle and eigenprojections $P_j$, and $\cL$ is any Lindbladian \cite{Gorini.1976,Lindblad.1976}. If $J = 1$, this reduces to the case discussed in \cite{Moebus.2019}, and if $M$ is a projection and $\cL$ a Hamiltonian to the case discussed in the founding papers \cite{Beskow.1967,Misra.1977}. Under appropriate assumptions, the above limit can be generalized to infinite dimensions and unbounded generators \cite{Moebus.2024,Moebus.2022,Becker.2021}, which is also discussed in this work.

    In this work, we improve \Cref{eq:Zeno} by summing over differently weighted product formulas:
    \begin{equation*}
        \begin{aligned}
            \Bigl\|\sum_{\ell=1}^{K+1} c_\ell \bigl(Me^{\frac{t}{\ell n p}\cL}\bigr)^{\ell n p} - \sum_{j=1}^{J} \lambda_j e^{tP_j\cL P_j} P_j \Bigr\|_{1 \rightarrow 1}
            &= \cO\Bigl(\frac{1}{n^{K+1}}\Bigr)\,.
        \end{aligned}
    \end{equation*}
	Since the coefficients $c_\ell$ are not necessarily positive, the result implies an exponential improvement in the convergence of Zeno-like expectations due to post-processing (see \Cref{thm:multi-product-zeno-effect-1} for details):
    \begin{equation*}
        \begin{aligned}
            \biggl|\sum_{\ell=1}^{K+1} c_\ell\tr\bigl[O\bigl(Me^{\frac{t}{\ell n p}\cL}\bigr)^{\ell n p}\bigr]- \tr\bigl[O\sum_{j=1}^{J} \lambda_j e^{tP_j\cL P_j} P_j(\rho)\bigr]\biggr|&= \cO\Bigl(\frac{1}{n^{K+1}}\Bigr)\,.
        \end{aligned}
    \end{equation*}
	
	The insight that a quantum dynamical semigroup can be projected onto a predefined subspace through frequent `kicks' or `pulses' is fundamental for decoupling and motivates the use of the quantum Zeno effect in error correction schemes to mitigate decoherence in open quantum systems \cite{Hosten.2006,Franson.2004,Beige.2000,Barankai.2018,Luchnikov.2017,Facchi.2004,Dominy.2013,Erez.2004}. In particular, the generalized quantum Zeno effect is equivalent to various `Bang-Bang' methods \cite{Facchi.2004,Burgarth.2022,Burgath.2024}, which can also be used to engineer an evolution by sequentially interrupting dissipative dynamics \cite{Viola.1998,Viola.2002}. Other applications of the quantum Zeno effect include ground state preparation of a two-dimensional quantum many-body Hamiltonian \cite{Zhang.2024}, purifying quantum states \cite{Yuasa.2003}, generating entanglement \cite{Nakazato.2004,Wang.2008}, and engineering the $X(\theta)$ gate in the bosonic cat-code \cite{Guillaud.2023,Guillaud.2019}, with more examples found in the literature.
	
	In terms of simulating the above systems, the speed of convergence of the quantum Zeno effect is quite limited. In comparison, Trotter's product formula was improved in the early days by Suzuki's work \cite{Suzuki.1997}, demonstrating that a convergence rate of order $\cO\left(\tfrac{1}{n^{K}}\right)$ for arbitrary $K \in \mathbb{N}$ can be achieved. Unfortunately, it is not clear whether --- and if so, how --- this approach applies to the Zeno effect. Moreover, Trotter's product formula has seen further advances in simulation through so-called multi-product formulas, which improve the gate complexity in comparison to the Suzuki approach \cite{Childs.2021,Low.2019,Zhuk.2024,Aftab.2024,Mizuta.2026,Carrera.2023,Yu.2025}. The analysis of these formulas heavily relies on the Baker-Campbell-Hausdorff formula, which is not applicable to the quantum Zeno effect.
    
	Nevertheless, in the works \cite{Moebus.2019,Moebus.2022,Moebus.2024}, we provided proof strategies that omit counting arguments to address the unbounded case. These strategies allow for a simple and explicit error representation, which is key for the multi-product Zeno effect --- the main result of this work.
    
	After introducing the mathematical and quantum mechanical framework in \Cref{sec:preliminaries}, we analyze the higher-order expansion of the quantum Zeno effect errors in \Cref{sec:higher-order-expansion}. In this section, all proofs are formulated for contractions on Banach spaces, in contrast to the subsequent \Cref{sec:multi-product}, which discusses the multi-product Zeno effect in the quantum mechanical setting, along with examples. Finally, the results are summarised in \Cref{sec:conclusion}, and all proofs are provided in the appendix.

\section{Preliminaries}\label{sec:preliminaries}
	Before presenting the results, we introduce the mathematical framework and notation. Throughout, $(\cX, \|\cdot\|_{\cX})$ denotes a Banach space and $\cB(\cX)$ the Banach space of bounded linear endomorphisms (operators) on $\cX$ with operator norm $\|A\|_{\cX \rightarrow \cX} \coloneqq \sup_{\|x\|_{\cX} = 1} \|A(x)\|_{\cX}$ for $A \in \cB(\cX)$ and identity $\1$. We call an operator $A \in \cB(\cX)$ a projection if $A^2 = A$ and a contraction if $\|A\|_{\cX \rightarrow \cX} \leq 1$. The resolvent set of $A \in \cB(\cX)$ is defined by $\mathrm{res}(A) \coloneqq \{z \in \C \;|\; (z - A) \text{ is bijective}\}$, its spectrum by $\spec(A) \coloneqq \C \backslash \mathrm{res}(A)$, and the resolvent of $A$ is defined by $R(z, A) \coloneqq (z - A)^{-1}$ for all $z \in \mathrm{res}(A)$.
	
	Next, an operator-valued map $I \ni t \mapsto X(t) \in \cB(\cX)$ defined on an interval $I \subset \R$ is continuous if it is continuous with respect to the operator norm, strongly continuous if the vector-valued map $I \ni t \mapsto X(t)x \in \cX$ is continuous w.r.t.~the norm on $\cX$, and differentiable at $t \in I$ if it is Fréchet differentiable at $t$ with respect to the operator norm. If the operator-valued map is (strongly) continuous on $[a, b] \subset I$, the Bochner integral $\int_{a}^b X(t) dt$ (resp.~$\int_{a}^b X(t)x dt$ for $x\in\cX$) is well-defined and extends the Lebesgue integral to operator-valued maps. It satisfies the usual properties of linearity, the triangle inequality, additivity of disjoint sets, dominated convergence, and the fundamental theorem of calculus can be proven:
	\begin{equation*}
		X(b) - X(a) = \int_{a}^{b} \frac{d}{dt} X(t) dt\,,
	\end{equation*}
	assuming that $X(t)$ is differentiable for all $t \in (a, b)$ with continuous derivative \cite[Thm.~3.7.4]{Hille.1996}. Next, a strongly continuous ($C_0$-)semigroup is defined by an operator-valued map $(T_t)_{t \geq 0}$ that satisfies the following properties for all $t, s \geq 0$ and $x \in \cX$:
	\begin{equation*}
		T_t T_s = T_{t+s}, \quad T_0 = \1, \quad \text{and} \quad \lim_{t \downarrow 0} T_t x = x\,.
	\end{equation*}
	Note that the latter continuity condition is in the strong operator topology, i.e., pointwise convergence for $x \in \cX$, and not in the uniform operator topology with respect to the operator norm. Every $C_0-$semigroup is uniquely determined by its generator given by
	\begin{equation*}
		\cL(x) = \lim_{t \to 0^+} \frac{1}{t}(T_t(x) - x) \,,
	\end{equation*}
	for any $x \in \cD(\cL) = \{x \in \cX : t \mapsto T_t(x) \text{ is differentiable on } \R_{\geq 0}\}$, which is an unbounded, densely defined, and closed operator, justifying the notation $T_t = e^{t \cL}$. An unbounded, densely defined operator is a linear map $\cL:\cD(\cL) \subset \cX \rightarrow \cX$ defined on a domain $\cD(\cL)$ dense in $\cX$. Moreover, it is called closed if its graph $\{(x, \cL(x))\,|\,x \in \cD(\cL)\}$ is a closed set in the product space $\cX \times \cX$. The generator is bounded if and only if $(T_t)_{t \geq 0}$ is continuous in the uniform topology \cite[Thm.~II.1.4]{Engel.2000}. By convention, we extend all densely defined and bounded operators by the bounded linear extension theorem to bounded operators on $\cX$ \cite[Thm.~2.7-11]{Kreyszig.1989}.
	
	Quantum mechanical systems are defined on a separable Hilbert space $\cH$ over the complex field $\C$ with an inner product $\braket{\cdot\,|\,\cdot}$ and induced norm $\|\psi\| \coloneqq \sqrt{\abs{\braket{\psi\,|\,\psi}}}$. The adjoint $A^\dagger$ of $A \in \cB(\cH)$ is uniquely defined by $\braket{A\phi\,|\,\varphi} = \braket{\phi\,|\,A^\dagger \varphi}$ for all $\phi, \varphi \in \cH$, and $A$ is positive if $A = B^\dagger B$ for some $B \in \cB(\cH)$. Then, the trace norm of $A \in \cB(\cH)$ is defined by $\|A\|_1 = \tr[|A|]$ with absolute value $|A| \coloneqq \sqrt{A^\dagger A}$ \cite[Thm.~2.4.4]{Simon.2015} and trace $\tr[|A|] = \sum_{j \in \N} \bra{\varphi_j} |A| \ket{\varphi_j}$ for any orthonormal basis $\{\ket{\varphi_j}\}_{j \in \N}$. The Banach space of trace-class operators is denoted by $\cT_1 = \{A \in \cB(\cH) \,|\, \|A\|_1 < \infty\}$ \cite[Sec.~3.6]{Simon.2015}. The subset $S_1 \subset \cT_1$ of positive operators with trace $1$ is called density operators or quantum states, which evolve via quantum operations in the Schrödinger picture. In the Schrödinger picture, a quantum operation $\cE\in\cB(\cT_1)$ is a completely positive, i.e.~$\cE\otimes\1(X)\geq0$ for all $X\geq0$, and trace non-increasing, i.e.~$\tr[\cE(X)]\leq\tr[X]$ for all $X\geq0$, linear map, and a quantum channel if equality holds in the latter. Then, a $C_0$-semigroup defined on $\cX = \cT_1$ is called a quantum dynamical $C_0$-semigroup if it is a quantum operation as well. To connect the first operator theoretical part and the quantum mechanical part of this section, note that every quantum operation is a contraction on $\cX=\cT_1$ with respect to the operator norm $\|\cdot\|_{1\rightarrow 1}$ (or diamond norm) \cite{WolfLecture}.

\section{Higher order expansion}\label{sec:higher-order-expansion}
	The main aim of this section is to analyze the higher-order terms of the quantum Zeno effect error in powers of $\frac{1}{n}$. Specifically, we expand the following difference in terms of $\frac{1}{n}$:
	\begin{equation}\label{eq:Zeno-error}
		\left( M e^{\frac{1}{n}t \mathcal{L}} \right)^n - \sum_{j=1}^{J} \lambda_j^n e^{t P_j \mathcal{L} P_j} P_j\,.
	\end{equation}
	Recall that in this section we consider contractions, contractive $C_0$-semigroups, and contractive operator-valued maps on a Banach space $\mathcal{X}$. This framework will be applied to the quantum setting in \Cref{sec:multi-product}, where $\mathcal{X}$ is chosen as the space of trace-class operators, as discussed in \Cref{sec:multi-product}.

	\subsection{Projective quantum Zeno effect}\label{subsec:projective-zeno}
		To provide a clear progression of ideas, we first introduce the tools used to prove an expansion for the projective Zeno effect (where $M = P$, $J = 1$, and $\lambda_1 = 1$).
		
		\subsubsection{Auxiliary lemmas}
			The proof relies on a modification of Chernoff's $\sqrt{n}$-Lemma \cite{Chernoff.1968} and the Dunford–Segal approximation \cite{Dunford.1946, Gomilko.2014}. We begin with the former:
			
			\begin{lem}[Modified Chernoff Lemma \texorpdfstring{\cite[Lem.~4.2]{Moebus.2022}}{???}]\label{lem:chernoff}
				Let $C \in \cB(\cX)$ be a contraction on a Banach space $\cX$. Then, $(e^{t(C - \1)})_{t \geq 0}$ is a contraction semigroup, and
				\begin{equation*}
					\left(C^n - e^{n(C - \1)}\right) = -n(C - \1)^2 \int_{0}^{1} \tau C_\tau^{n-1} e^{(1-\tau) n(C - \1)}\, d\tau\,,
				\end{equation*}
				where $C_{\tau} = (1-\tau)\1 + \tau C$. If $C: [0,1] \rightarrow \cB(\cX),\, s \mapsto C(s)$ is a contractive, continuously differentiable operator-valued map with $C(0)=\1$, then for $s \in [0,1]$,
				\begin{align*}
					\Bigl(C^n(s) - e^{n(C(s) - \1)}\Bigr) &= -ns^2 \iiint_{0}^{1} \tau_1 e^{(1-\tau_1) n(C(s) - \1)} C'(s\tau_2) C'(s\tau_3) C_{\tau_1}^{n-1}(s) \, d\tau_{321}\,,
				\end{align*}
				where $C'(s)$ is the first-order Fréchet derivative, $C^n(s)=(C(s))^n$, and $d\tau_{321}$ abbreviates $d\tau_{3} d\tau_{2} d\tau_{1}$.
			\end{lem}
			
			\begin{proof}
				The proof is given in detail in \Cref{appx-sec:projective-zeno-auxiliary}.
			\end{proof}
			
			Next, we prove a modification of the Dunford–Segal approximation \cite{Dunford.1946, Gomilko.2014}:
			
			\begin{lem}[Dunford–Segal approximation]\label{lem:dunford-segal}
				Let $C: [0,1] \rightarrow \cB(\cX),\, s \mapsto C(s)$ be a twice continuously differentiable operator-valued map with $C(0)=\1$. Then, for $s_1,s_2,t\in[0,1]$,
				\begin{equation*}
					\Bigl(e^{\frac{t}{s_1}(C(s_2) - \1)} - e^{tC'}\Bigr) = t \frac{s_2^2}{s_1} \iiint_0^1 e^{(1 - \tau_1)\frac{t}{s_1}(C(s_2) - \1)} \biggl(\tau_2 C''(\tau_3\tau_2 s_2) + \frac{s_2-s_1}{s_2^2} C' \biggr) e^{\tau_1tC'} \, d\tau_{321}
				\end{equation*}
				where $C'(s)$ and $C''(s)$ are the first and second-order Fréchet derivatives, and $C' = C'(0)$. If $C(s)$ is a contraction for all $s\in[0,1]$, then $(e^{tC'})_{t\geq0}$ is a contraction semigroup.
			\end{lem}
			\begin{proof}
				The proof is given in detail in \Cref{appx-sec:projective-zeno-auxiliary}.
			\end{proof}
		
		\subsubsection{Higher-order expansion of the projective Zeno effect}
			With the help of the auxiliary result above, we prove an explicit representation of the first-order projective Zeno error.
			\begin{prop}[Zeno second-order expansion]\label{prop:projective-zeno-second-order}
				Let $(\cL, \cD(\cL))$ be the generator of a contractive $C_0$-semigroup, and let $P \in \cB(\cX)$ be a contractive projection such that $P\cL$, $\cL P$, and $P\cL^2$ are bounded. Then, for all $n \geq 1$ and $t \geq 0$, 
				\begingroup\small
				\begin{equation*}
					\begin{aligned}
						&\Bigl\|(Pe^{\frac{t}{n}\cL})^n - e^{tP\cL P}P - \frac{t^2}{2n}\int_0^1 e^{\tau_1 P\cL P} P\cL[\cL, P] P e^{t(1 - \tau_1)P\cL P} d\tau_1- \frac{t}{n}e^{tP\cL P}P\cL(\1-P) \Bigr\|_{\cX \rightarrow \cX}\\
						&\qquad\qquad\qquad\qquad\qquad\leq\frac{1}{2n^2}\Bigl(1+2t^2\|P\cL^2\|_{\cX\rightarrow\cX}+t\|\cL P\|_{\cX\rightarrow\cX}(4+t\|\cL P\|_{\cX\rightarrow\cX})\Bigr)^2\\
					\end{aligned}
				\end{equation*}
				\endgroup
			\end{prop}
			\begin{proof}[Proof-sketch]
				A detailed proof can be found in \Cref{appx-sec:projective-zeno-main}. In the first step, we apply the modified Chernoff Lemma \eqref{lem:chernoff} and the Dunford-Segal approximation \eqref{lem:dunford-segal} to achieve the zero-order approximation in a slightly generalized setup. By applying the above tools a second time to the error explicitly given by an integral, we derive the statement.
			\end{proof}
			\begin{rmk}[Bounded generators]
				Note that \Cref{prop:projective-zeno-second-order} applies directly to bounded $\cL$, and the assumptions that $P\cL$, $\cL P$, and $P\cL^2$ are bounded become redundant.
			\end{rmk}
			\begin{rmk}[Time dependency]
				In the above result, the convergence rate is also tracked in time and is shown to be a quartic polynomial. Due to the inductive nature of the proof and the paper's focus on higher-order expansions, we do not track the time-dependence in the following and keep it constant. However, it is an interesting question to investigate in future work whether the time dependency can be improved.
			\end{rmk}
			In the following main result for the projective Zeno effect, we show that an expansion in powers of $\frac{1}{n}$ exists for every $K\in\N$, provided that the appropriate boundedness conditions are satisfied.
			\begin{thm}[Projective Zeno higher-order expansion]\label{thm:projective-zeno-higher-order}
				Let $(\cL, \cD(\cL))$ be the generator of a contractive $C_0$-semigroup, and $P \in \cB(\cX)$ be a contractive projection such that $P\cL^{k}$ and $\cL P$ are bounded for all $k\in\{1,...,K+1\}$. Then, there is a sequence $\{E_{k}\}_{k=1}^K \subset \cB(\cX)$, independent of $n$, such that
				\begin{equation*}
					\Bigl\|(Pe^{\frac{t}{n}\cL})^n - e^{tP\cL P}P - \sum_{k=1}^{K} \frac{1}{n^k} E_k\Bigr\|_{\cX \rightarrow \cX} = \cO\Bigl(\frac{1}{n^{K+1}}\Bigr)\,.
				\end{equation*}
			\end{thm}		
			\begin{proof}[Proof-sketch]
				A detailed proof can be found in \Cref{appx-sec:projective-zeno-main}. The proof is based on induction over $K$. \Cref{prop:projective-zeno-second-order} provides the `induction start' for $K=1$, and then it is shown that the procedure used in the proof of \Cref{prop:projective-zeno-second-order} can be iterated under certain boundedness assumptions.
			\end{proof}
			
			To clarify the relation to the bounded generator case, we briefly present a simplified version.
			\begin{cor}[Bounded generator]\label{cor:projective-zeno-higher-order-bounded}
				Let $\cL \in \cB(\cX)$ generate a contractive semigroup, and $P \in \cB(\cX)$ be a contractive projection. Then, there is a sequence of operators $\{E_{k}\}_{k=1}^K \subset \cB(\cX)$, independent of $n$, such that for all $n \geq 1$ and $t \geq 0$,
				\begin{equation*}
					\Bigl\|(Pe^{\frac{t}{n}\cL})^n - e^{tP\cL P}P - \sum_{k=1}^{K} \frac{1}{n^k} E_k\Bigr\|_{\cX \rightarrow \cX} = \cO\Bigl(\frac{1}{n^{K+1}}\Bigr)\,.
				\end{equation*}
			\end{cor}

	\subsection{Generalized quantum Zeno effect}\label{subsec:generalized-zeno}
		With the expansion of the projective Zeno effect at hand (\Cref{thm:projective-zeno-higher-order}), we can use the holomorphic functional calculus \cite[Thm.~2.3.1-3]{Simon.2015} to extend the result to the generalized quantum Zeno effect. Before presenting the auxiliary result, we first discuss the necessary assumptions: We assume that $M\in\cB(\cX)$ satisfies the uniform power convergence assumption:
		\begin{equation}\label{eq:power-convergence}
			\Bigl\|M^n - \sum_{j=1}^J \lambda_j^n P_j \Bigr\|_{\cX \rightarrow \cX} \leq c \delta^n
		\end{equation}
		for a constant $c \geq 0$, $n\in\N$, peripheral eigenvalues $\{\lambda_j\}_{j=1}^J\subset\partial\D_1=\{z\in\C\,|\,|z|=1\}$ on the unit sphere, and eigenprojections $\{P_j\}_{j=1}^J$. The labeling is motivated by the fact that the power convergence above is equivalent to the spectral gap condition \cite[Prop.~3.1]{Becker.2021}:
			\begin{equation}\label{eq:spectral-gap}
				\spec(M) \subset \D_{\delta} \cup \{\lambda_1,...,\lambda_J\}
			\end{equation}
			with the ball $\D_{\delta}=\{z\in\C\,|\,|z|\leq\delta\}$ and eigenprojections defined by the holomorphic functional calculus:
			\begin{equation}\label{eq:eigenprojections}
				P_j = \frac{1}{2\pi i} \oint_{\Gamma_j} R(z, M) dz
			\end{equation}
			and zero nilpotent parts $N_j = (\lambda_j - M) P_j$. Here, $\Gamma_j$ is the curve 
			\begin{equation*}
				\Gamma_j:[0,1]\rightarrow\C,\Gamma_j(s)=\lambda_j+\min_{i\in\{1,...,J\}}\Bigl\{\frac{1-\delta}{3},\frac{|\lambda_i-\lambda_j|}{3}\Bigr\}e^{2\pi i s}\,,
			\end{equation*}
			which separates the peripheral eigenvalues $\lambda_j$ from the rest of the spectrum.
		\begin{figure}[t]
			\centering
				\begin{tikzpicture}
					\draw[dashed,tumivory] (0,0) ellipse (1.4cm and 1.4cm);
					\draw[dashed,mygreen] (0,0) ellipse (1.15cm and 1.15cm);
					\filldraw (0,0) ellipse (0.03cm and 0.03cm);
					
					\filldraw (15:1.4) ellipse (0.03cm and 0.03cm);
					\draw[dashed,myorange] (15:1.4) ellipse (0.2cm and 0.2cm);
					\draw[myorange] (15:1.9) node {$\Gamma_1$};
					
					\filldraw (70:1.4) ellipse (0.03cm and 0.03cm);
					\draw[dashed,myorange] (70:1.4) ellipse (0.2cm and 0.2cm);
					\draw[myorange] (70:1.9) node {$\Gamma_2$};
					
					\filldraw (110:1.4) ellipse (0.03cm and 0.03cm);
					\draw[dashed,myorange] (110:1.4) ellipse (0.2cm and 0.2cm);
					\draw[myorange] (110:1.9) node {$\Gamma_3$};
					
					\filldraw (130:1.4) ellipse (0.03cm and 0.03cm);
					\draw[dashed,myorange] (130:1.4) ellipse (0.2cm and 0.2cm);
					\draw[myorange] (130:1.9) node {$\Gamma_4$};
					
					\filldraw (170:1.4) ellipse (0.03cm and 0.03cm);
					\draw[dashed,myorange] (170:1.4) ellipse (0.2cm and 0.2cm);
					\draw[myorange] (170:1.9) node {$\Gamma_5$};
					
					\filldraw (195:1.4) ellipse (0.03cm and 0.03cm);
					\draw[dashed,myorange] (195:1.4) ellipse (0.2cm and 0.2cm);
					\draw[myorange] (195:1.9) node {$\Gamma_6$};
					
					\filldraw (260:1.4) ellipse (0.03cm and 0.03cm);
					\draw[dashed,myorange] (260:1.4) ellipse (0.2cm and 0.2cm);
					\draw[myorange] (260:1.9) node {$\Gamma_7$};
					
					\filldraw (300:1.4) ellipse (0.03cm and 0.03cm);
					\draw[dashed,myorange] (300:1.4) ellipse (0.2cm and 0.2cm);
					\draw[myorange] (300:1.9) node {$\Gamma_8$};
					
					\filldraw (350:1.4) ellipse (0.03cm and 0.03cm);
					\draw[dashed,myorange] (350:1.4) ellipse (0.2cm and 0.2cm);
					\draw[myorange] (350:1.9) node {$\Gamma_J$};
					
					\draw [dotted,myorange,domain=310:340] plot ({1.9*cos(\x)}, {1.9*sin(\x)});
					
					\draw[<->] (0.0,0.05)--(0,0.95);
					\draw[<->,mygreen] (0.0433,-0.025)--(0.9526cm,-0.55cm);
					
					\fill[pattern=my north east lines] (0,0) ellipse (1cm and 1cm);
					\filldraw[white] (-0.35,0.5) ellipse (0.15cm and 0.25cm);
					\draw (-0.35,0.5) node {$\delta$};
					\filldraw[white] (0.35,-0.5) ellipse (0.15cm and 0.25cm);
					\draw[mygreen] (0.35,-0.5) node {$\tilde{\delta}$};
					\filldraw[white] (-1.06,-1.06) ellipse (0.15cm and 0.3cm);
					\draw (-1.06,-1.06) node [mygreen]{$\gamma$};
					
					\draw (2,1.2) node {$\spec(M)$};
				\end{tikzpicture}
			\caption{Spectral separation of the peripheral eigenvalues from the strictly contractive part of the spectrum; adapted from \cite[Fig.~3]{Moebus.2022}.}\label{fig1}
		\end{figure}
		We also assume that $M\cL^k$ and $\cL M$ are bounded for all $k \in \{1, \dots, K\}$. Under the above assumptions, \cite[Thm.~6.1]{Moebus.2022} provides the following convergence:
		\begin{equation}\label{eq:generalized-Zeno-effect}
			\Bigl\|(Me^{\frac{t}{n}\cL})^n - \sum_{j=1}^{J}\lambda_j^{n}e^{tP_j\cL P_j}P_j\Bigr\|_{\cX \rightarrow \cX} = \cO\Bigl(\frac{1}{n}\Bigr)\,.
		\end{equation}
		Note that the projections $P_j$ are, by assumption, contractive. This follows from the average operator defined in Equation 4.3 of \cite{Becker.2021}, which is contractive, and thus, its limit $P_j$ inherits this property.
		
		\subsubsection{Auxiliary Lemma}
			Next, we present a description of the perturbed eigenprojection of an isolated part of the spectrum using the holomorphic functional calculus (cf.~\cites[Lem.~6.2]{Moebus.2022}[Lem.~5.3]{Becker.2021}).
			\begin{lem}\label{lem:functional-calculus}
				Let $M \in \cB(\cX)$ satisfy the spectral gap assumption \eqref{eq:spectral-gap}, and let $\Gamma$ be a simple, closed, rectifiable curve\footnote{A simple, closed, rectifiable curve is a continuous, non-intersecting loop with finite length.} that separates, but does not intersect, the gapped spectrum. Let $(\cL, \cD(\cL))$ be the generator of a contractive $C_0$-semigroup, with $M\cL^{k}$ bounded for all $k \in \{1, \dots, K+1\}$. Then, there exists an $\epsilon > 0$ such that for all $t \in [0, \epsilon]$,
				\begin{equation*}
					P(t) \coloneqq \frac{1}{2\pi i} \oint_{\Gamma}R(z,Me^{t\cL})\,dz
				\end{equation*}
				defines a projection onto the eigenspace enclosed by $\Gamma$. Moreover, there exists a sequence of bounded operators $\{F_k\}_{k=1}^K \subset \cB(\cX)$ such that for $t \rightarrow 0$
				\begin{equation*}
					\Bigl\|P(t) - P(0) - \sum_{k=1}^{K}t^kF_k\Bigr\|_{\cX\rightarrow \cX} = \cO(t^{K+1})\,.
				\end{equation*}
			\end{lem}
			\begin{proof}
				A detailed proof can be found in \Cref{appx-sec:generalized-zeno-auxiliary}.
			\end{proof}
		
		\subsubsection{Higher order expansion of the quantum Zeno effect}
			With the help of the above result, we can follow the proof strategy of \cite{Moebus.2019} to establish the following result.
			\begin{thm}[Generalized Zeno higher-order expansion]\label{thm:generalized-zeno}
				Let $(\cL, \cD(\cL))$ be the generator of a contractive $C_0$-semigroup, and let $M \in \cB(\cX)$ be a contraction satisfying the assumption of uniform power convergence \eqref{eq:power-convergence}. Assume that $M\cL^{k}$ and $\cL M$ are bounded for all $k \in \{1, \dots, K+1\}$. Then, there exists a sequence of operators $\{E_{k,j}\}_{k,j=1}^{K,J} \subset \cB(\cX)$, independent of $n$, such that
				\begin{equation*}
					\Bigl\|(Me^{\frac{t}{n}\cL})^n - \sum_{j=1}^{J}\lambda_j^{n}e^{tP_j\cL P_j}P_j - \sum_{k=1}^{K} \frac{1}{n^k} \sum_{j=1}^{J}\lambda_j^nE_{k,j}\Bigr\|_{\cX \rightarrow \cX} = \cO\Bigl(\frac{1}{n^{K+1}}\Bigr)\,.
				\end{equation*}
			\end{thm}
			\begin{proof}[Proof-sketch]
				A detailed proof can be found in \Cref{appx-sec:generalized-zeno-main}. The proof uses the holomorphic functional calculus to define the perturbed eigenprojections on the unit circle (see \Cref{lem:functional-calculus}). This is then used to decompose the Zeno sequence into a part related to the spectrum in $\D_{\delta}$, which converges to $0$, and another part, for which we follow the proof of \Cref{thm:projective-zeno-higher-order}. Finally, the expansion of the perturbed eigenprojection is used to conclude the proof.
			\end{proof}

\section{Multi-product Zeno effect}\label{sec:multi-product}
	With the above proven higher order error expansion for the generalized quantum Zeno effect, we can construct a multi-product formula for the quantum Zeno effect similar to the multi-product trotter formulas \cite{Childs.2021,Childs.2021}. Due to the strong quantum mechanical motivation, we construct the Multi-product Zeno effect for quantum operation $M$ and quantum dynamical $C_0$-semigroups $e^{t\cL}$ defined on trace-class operators $\cT_1$ of an separable Hilbert space, which are by definition contractions. However, the the following also works for operators, as discussed in \Cref{sec:higher-order-expansion}.  
	
	\subsection{Multi-product scheme}\label{subsec:multi-product-scheme}
		For a given vector of real values $c \in \R^{K+1}$, we define the multi-product Zeno effect by
		\begin{equation}\label{eq:multi-product-formula}
			\sum_{\ell=1}^{K+1}c_\ell\left(Me^{\frac{t}{n_\ell}\cL}\right)^{n_\ell}\,.
		\end{equation}
		The goal is to choose the coefficients such that the higher-order error terms found in \Cref{thm:generalized-zeno} cancel out. This is first demonstrated using the first-order expansion of the projective Zeno effect.
		
		\subsubsection{Projective $2$-product Zeno effect}
			Recall that \Cref{prop:projective-zeno-second-order} proves that for any projective quantum operation $P$, with $\overline{P} = (\1 - P)$, and a quantum dynamical $C_0$-semigroup $(e^{t\cL})_{t \geq 0}$, where $P\cL$, $\cL P$, and $P\cL^2$ are bounded, there is an explicit bound such that for $n \to \infty$ and $t \geq 0$
			\begin{equation*}
				\begin{aligned}
					&\left\|(Pe^{\frac{t}{n}\cL})^n - e^{tP\cL P}P
					- \frac{t^2}{2n}\int_0^1 e^{\tau_1 P\cL P} P\cL[\cL, P] P e^{t(1 - \tau_1)P\cL P} d\tau_1\right.\\
					&\qquad\left.{}- \frac{t}{n}e^{tP\cL P}P\cL\overline{P} \right\|_{\cX \rightarrow \cX}
					= \cO\left(\frac{1}{n^2}\right)\,.
				\end{aligned}
			\end{equation*}
			Then, we search for a $c\in\R^2$ such that the multi-product formula for $K = 1$ and $n_\ell = \ell n$ satisfies
			\begin{equation*}
				\sum_{\ell=1}^{2}c_\ell(Pe^{\frac{t}{n\ell}\cL})^{n\ell} - c_\ell e^{tP\cL P}P - c_\ell\frac{1}{n\ell}E_1 = \cO\left(\frac{1}{n^2}\right)
			\end{equation*}
			where
			\begin{equation*}
				E_1 \coloneqq \frac{t^2}{2}\int_0^1 e^{\tau_1 P\cL P} P\cL[\cL, P] P e^{t(1 - \tau_1)P\cL P} d\tau_1 - te^{tP\cL P}P\cL\overline{P}\,.
			\end{equation*}
			This is satisfied if the zeroth order is unchanged and the first-order error cancels or equivalently if the following system of linear equations is satisfied:
			\begin{equation*}
				\begin{pmatrix}
					1 & 1 \\
					\frac{1}{1} & \frac{1}{2}
				\end{pmatrix}
				\begin{pmatrix}
					c_1 \\
					c_2
				\end{pmatrix}
				=
				\begin{pmatrix}
					1 \\
					0
				\end{pmatrix}
			\end{equation*}
			which is solved by $c_1 = -1$ and $c_2 = 2$. This gives the following $2$-product scheme:

			\begin{cor}[2-product Zeno effect]
				Let $(\cL,\cD(\cL))$ be the generator of a quantum dynamical $C_0$-semigroup, $P \in \cB(\cT_1)$ be any projective quantum operation such that $P\cL, P\cL^2, \cL P \in \cB(\cX)$ are bounded, $O \in \cB(\cH)$ be an observable, and $\rho \in \cT_1(\cH)$. Then for all $n \geq 1$ and $t \geq 0$, we have
				\begin{equation*}
					\begin{aligned}
						\Bigl\|2(Pe^{\frac{t}{2n}\cL}P)^{2n} - (Pe^{\frac{t}{n}\cL}P)^n - e^{tP\cL P}P\Bigr\|_{1\rightarrow 1} &\leq \cO\left(\frac{1}{n^2}\right)\,,
					\end{aligned}
				\end{equation*}
				in particular,
				\begin{equation*}
					\begin{aligned}
						\Bigl|2\tr[O(Pe^{\frac{t}{2n}\cL}P)^{2n}(\rho)] - \tr[O(Pe^{\frac{t}{n}\cL}P)^n(\rho)] - \tr[Oe^{tP\cL P}P(\rho)]\Bigr| &\leq \cO\left(\frac{1}{n^2}\right)\,.
					\end{aligned}
				\end{equation*}
			\end{cor}
			This means that running two Zeno effects in parallel with different step sizes can improve the convergence rate to $\frac{1}{n^2}$, while only linearly increasing the number of channel usages.
		
		\subsubsection{Generalized Multi-product Zeno effect}
			In the following, we generalize the above scheme, similar to the multi-product Trotter formulas discussed, for example, in \cite{Childs.2021, Carrera.2023}, using the higher-order expansion of the generalized quantum Zeno effect (see \Cref{thm:generalized-zeno}). That is, there exists a sequence of bounded operators $\{E_k\}_{k=1}^K$ such that
			\begin{equation*}
				\Bigl\|(Me^{\frac{t}{n}\cL})^n - \sum_{j=1}^{J}\lambda_j^{n}e^{tP_j\cL P_j}P_j - \sum_{k=1}^{K} \frac{1}{n^k} \sum_{j=1}^{J}\lambda_j^nE_{k,j}\Bigr\|_{1 \rightarrow 1} = \cO\left(\frac{1}{n^{K+1}}\right)\,.
			\end{equation*}
			For a quantum operator $M$ admitting the uniform power convergence assumption 
			\begin{equation*}
				\Bigl\|M^n - \sum_{j=1}^J \lambda_j^n P_j \Bigr\|_{\cX \rightarrow \cX} \leq c \delta^n
			\end{equation*}
			given a  $c \geq 0$, $n \in \N$, peripheral  eigenvalues $\{\lambda_j\}_{j=1}^J \subset \partial \D_1$, and a quantum dynamical $C_0$-semigroup with generator $(\cL,\cD(\cL))$, so that $M\cL^k$, $\cL M$ are bounded for all $k \in \{1, \dots, K+1\}$. Additionally, assume that the peripheral  eigenvalues are periodic, i.e., the phases are rational:
			\begin{equation*}
				\{\lambda_j\}_{j=1}^{J}\subset\{e^{2\pi i q}\,|\,q\in\Q\}\,.
			\end{equation*}
			Therefore, there exists a common period $p$ such that $\lambda_j^p=\lambda_j$ for all $j\in\{1,...,J\}$. Then, choosing $n_\ell=m_\ell p n$, we obtain
			\begin{equation}\label{eq:multi-product-error}
				\begin{aligned}
					&\left\|\sum_{\ell=1}^{K+1}c_{\ell}(Me^{\frac{t}{m_\ell pn}\cL})^{m_\ell pn}
					- \sum_{j=1}^{J}\lambda_j e^{tP_j\cL P_j}P_j\right.\\
					&\qquad\left.{}- \sum_{k=1}^{K}\frac{1}{(pn)^k}\Bigl(\sum_{\ell=1}^{K+1} \frac{c_\ell}{m_\ell^{k}}\Bigr)
					\sum_{j=1}^{J}\lambda_j E_{k,j}\right\|_{1 \rightarrow 1}
					= \cO\left(\frac{1}{n^{K+1}}\right)\,.
				\end{aligned}
			\end{equation}
			if 
			\begin{equation}\label{eq:multi-product-condition}
				\sum_{\ell=1}^{K+1}c_\ell = 1\quad\text{and}\quad\sum_{\ell=1}^{K+1}c_\ell m_\ell^{-k}=0\quad\text{for }k=1,...,K\,.
			\end{equation}
			For the choice $m_\ell=\ell$, the multi-product Zeno effect converges with a rate of order $\frac{1}{n^{K+1}}$, provided that the vector $c\in\R^{K+1}$ satisfies the following system of linear equations:
			\begin{equation}\label{LGS}
				\begin{pmatrix}
					1 & 1 & \hdots & 1\\
					\frac{1}{1} & \frac{1}{2} & \hdots & \frac{1}{K+1}\\
					\vdots & \vdots & \ddots & \vdots \\
					\frac{1}{1^{K+1}} & \frac{1}{2^{K+1}} & \hdots & \frac{1}{(K+1)^{K+1}}
				\end{pmatrix}
				\begin{pmatrix}
					c_1 \\
					c_2\\
					\vdots\\
					c_{K+1}
				\end{pmatrix}
				=
				\begin{pmatrix}
					1 \\
					0\\
					\vdots\\
					0
				\end{pmatrix}\,.
			\end{equation}
			Note that the system of linear equations in Eq.~\eqref{LGS} is a Vandermonde system and therefore invertible. However, this equidistant-grid choice is typically ill-conditioned for large $K$ (see \cite{Low.2019,Carrera.2023}). More importantly for the present setting, the coefficients grow exponentially: solving the Vandermonde system yields
			\begin{equation*}
				c_\ell = (-1)^{\ell-1}\binom{K+1}{\ell}\ell^K,
			\end{equation*}
			and thus $S_K\coloneqq\sum_{\ell=1}^{K+1}|c_\ell|$ satisfies
			\begin{equation}\label{eq:vandermonde}
				S_K \in \Omega(e^K),
			\end{equation}
			as discussed in \cite{Low.2019}. Since $S_K$ is part of the error bound, this limits the scalability of the multi-product Zeno effect for large $K$.
			\begin{cor}[Equidistant-grid case]\label{cor:vandermonde-case}
				Let $M$ be a quantum operation satisfying the uniform power convergence \eqref{eq:power-convergence} with periodic eigenvalues, i.e., $\{\lambda_j\}_{j=1}^{J}\subset\{e^{2\pi i q}\,|\,q\in\Q\}$, with period $p\in\N$, and let $(\cL,\cD(\cL))$ be the generator of a quantum dynamical $C_0$-semigroup satisfying $M\cL^k$ and $\cL M$ bounded for all $k\in\{1,...,K+1\}$. For $m_\ell=\ell\in\{1,...,K+1\}$, any $\varepsilon > 0$, all observables $O\in\cB(\cH)$, states $\rho\in\cT_1(\cH)$, and $\frac{t}{n}<\varepsilon$, we have
				\begin{equation*}
					\begin{aligned}
						\Bigl\|\sum_{\ell=1}^{K+1}c_\ell\bigl(Me^{\frac{t}{m_\ell n p}\cL}\bigr)^{m_\ell n p}-\sum_{j=1}^{J}\lambda_j e^{tP_j\cL P_j}P_j\Bigr\|_{1\rightarrow 1}&=\cO\Bigl(\frac{1}{n^{K+1}}\Bigr)\,,
					\end{aligned}
				\end{equation*}
				in particular,
				\begin{equation*}
					\begin{aligned}
						\tr\biggl[O\biggl(\sum_{\ell=1}^{K+1}c_\ell\bigl(Me^{\frac{t}{m_\ell n p}\cL}\bigr)^{m_\ell n p}-\sum_{j=1}^{J}\lambda_j e^{tP_j\cL P_j}P_j\biggr)(\rho)\biggr]&=\cO\Bigl(\frac{1}{n^{K+1}}\Bigr)\,.
					\end{aligned}
				\end{equation*}
				Here, the coefficients $c_\ell$ are given by $c_\ell = (-1)^{\ell-1}\binom{K+1}{\ell}\ell^K$ so that $\sum_{\ell=1}^{K+1}|c_\ell| \in \Omega(e^K)$ and $\sum_{\ell=1}^{K+1}m_\ell=\frac{(K+1)(K+2)}{2}$.
			\end{cor}

			\subsubsection{Well-conditioned multi-product grids}
			To overcome this limitation, we choose $m_\ell$ in $n_\ell \coloneqq m_\ell\, p\, n$ following the approach discussed in the Hamiltonian-simulation literature \cite{Childs.2021, Low.2019, Carrera.2023}. Recall that by \Cref{thm:generalized-zeno} each sequence satisfies the error bound in \Cref{eq:multi-product-error} under the condition \Cref{eq:multi-product-condition}, i.e.
			\begin{equation*}
				\sum_{\ell=1}^{K+1}c_\ell = 1\quad\text{and}\quad\sum_{\ell=1}^{K+1}c_\ell m_\ell^{-k}=0\quad\text{for }k=1,...,K\,.
			\end{equation*}
			is satisfied. This linear system has a $(K+1)\times(K+1)$ matrix determined by $\{m_\ell\}$ as observed in \Cref{LGS}. Next, we follow \cite{Low.2019} (see also \cite{Carrera.2023}) and choose $m_\ell$ as an approximate Chebyshev grid.

			\begin{thm}\label{thm:multi-product-zeno-effect-1}
				Let $M$ be a quantum operation satisfying the uniform power convergence \eqref{eq:power-convergence} with periodic eigenvalues, i.e., $\{\lambda_j\}_{j=1}^{J}\subset\{e^{2\pi i q}\,|\,q\in\Q\}$, with period $p\in\N$, and let $(\cL,\cD(\cL))$ be the generator of a quantum dynamical $C_0$-semigroup satisfying $M\cL^k$ and $\cL M$ bounded for all $k\in\{1,...,K+1\}$. Then for all $\varepsilon > 0$ and all observables $O\in\cB(\cH)$, states $\rho\in\cT_1(\cH)$, and $\frac{t}{n}<\varepsilon$, there exist pairwise distinct integers $m_1, \ldots, m_{K+1} \in \mathbb{N}$ and a coefficient vector $c \in \R^{K+1}$ solving the linear system \eqref{eq:multi-product-condition} such that the multi-product Zeno effect converges with rate $\frac{1}{n^{K+1}}$:
				\begin{equation*}
					\begin{aligned}
						\Bigl\|\sum_{\ell=1}^{K+1}c_\ell\bigl(Me^{\frac{t}{m_\ell n p}\cL}\bigr)^{m_\ell n p}-\sum_{j=1}^{J}\lambda_j e^{tP_j\cL P_j}P_j\Bigr\|_{1\rightarrow 1}&=\cO\Bigl(\frac{1}{n^{K+1}}\Bigr)\,,
					\end{aligned}
				\end{equation*}
				in particular,
				\begin{equation*}
					\begin{aligned}
						\tr\biggl[O\biggl(\sum_{\ell=1}^{K+1}c_\ell\bigl(Me^{\frac{t}{m_\ell n p}\cL}\bigr)^{m_\ell n p}-\sum_{j=1}^{J}\lambda_j e^{tP_j\cL P_j}P_j\biggr)(\rho)\biggr]&=\cO\Bigl(\frac{1}{n^{K+1}}\Bigr)\,.
					\end{aligned}
				\end{equation*}
			\end{thm}

			\begin{proof}
			First, recall that by \Cref{thm:generalized-zeno}, for each fixed $\ell\in\{1,\dots,K+1\}$ and $n_\ell=m_\ell pn$ (compare to \Cref{eq:multi-product-error}), there exist bounded operators $E_{k,j}$ (independent of $n$) and a remainder $R_{\ell,n}$ such that 
			\begin{equation*}
				\Bigl\|(Me^{\frac{t}{m_\ell pn}\cL})^{m_\ell pn} - \sum_{j=1}^{J}\lambda_je^{tP_j\cL P_j}P_j - \sum_{k=1}^{K}\frac{1}{(m_\ell pn)^k}\sum_{j=1}^{J}\lambda_jE_{k,j} - R_{\ell,n}\Bigr\|_{1 \rightarrow 1} = \cO\left(\frac{1}{n^{K+1}}\right)\,.
			\end{equation*}
			Then, summing over $\ell$ and using the condition \eqref{eq:multi-product-condition} yields
			\begin{equation*}
				\Bigl\|\sum_{\ell=1}^{K+1}c_{\ell}(Me^{\frac{t}{m_\ell pn}\cL})^{m_\ell pn} - \sum_{j=1}^{J}\lambda_j e^{tP_j\cL P_j}P_j\Bigr\|_{1 \rightarrow 1} = \cO\left(\frac{1}{n^{K+1}}\right)\,.
			\end{equation*}
			The argument above holds for every pairwise distinct integer grid $\{m_\ell\}$ for which \eqref{eq:multi-product-condition} is solved. It remains to explain, in detail, how to choose such a grid with good conditioning. For 
			\begin{equation*}
				x_\ell\coloneqq m_\ell^{-1}\in(0,1],\qquad \ell=1,\ldots,K+1\,,
			\end{equation*}
			\Cref{eq:multi-product-condition} is equivalent to the polynomial interpolation problem of finding $c_\ell$ such that
			\begin{equation*}
				\sum_{\ell=1}^{K+1} c_\ell x_\ell^{k}=\delta_{k0},\qquad k=0,1,\ldots,K.
			\end{equation*}
			In \Cref{eq:vandermonde}, we observed that the equidistant-grid choice $m_\ell=\ell$ yields exponentially growing $S_K$ (Vandermonde instability). The key idea of \cite[Thm.~1, Eqs.~(8-9)]{Low.2019} is to first relax $x_\ell$ to real values so that $\widetilde{m}_\ell=x_\ell^{-1}$ and start from a Chebyshev-type node family: The real nodes are
			\begin{equation*}
				x_j^{\star}=\sin^2\!\Bigl(\frac{\pi(2j-1)}{4(K+1)}\Bigr),\qquad j=1,\ldots,K+1,
			\end{equation*}
			These are the images of the Chebyshev-Gauss-Lobatto interpolation points; their use minimizes the Lebesgue constant and hence the polynomial interpolation error (see \cite[Chap.~5]{Trefethen.2019}). Solving the moment system on these nodes gives coefficients with logarithmic norm growth instead of exponential growth:
			\begin{equation*}
				\sum_{j=1}^{K+1}|\widetilde c_j|=\cO(\log(K+1))\qquad\text{and}\qquad\sum_{\ell=1}^{K+1}\widetilde{m}_\ell=\cO(K\log K)\,.
			\end{equation*}
			However, we require integer step counts, which is why we convert real exponents to integers by scale-and-round:
			\begin{equation*}
				m_\ell\coloneqq\lceil \kappa\widetilde m_\ell\rceil\in\N,
			\end{equation*}
			with $\kappa$ chosen so that the rounded values remain pairwise distinct (this is exactly the role of the scaling parameter in \cite[Eq.~(10)]{Low.2019}). The perturbation estimate of \cite[Eq.~(11)]{Low.2019} shows that rounding changes each coefficient by only a constant relative factor, provided the pre-rounded nodes are sufficiently separated. Therefore, one can choose pairwise distinct integers $m_1,\ldots,m_{K+1}$ such that \eqref{eq:multi-product-condition} is satisfied and
			\begin{equation*}
				S_K=\sum_{\ell=1}^{K+1}|c_\ell|=\cO(\log(K+1)).
			\end{equation*}
			This finishes the proof by choosing $n_\ell=m_\ell pn$.
		\end{proof}
			\begin{rmk}[Time dependency]
				From the proof of Theorem~\ref{thm:multi-product-zeno-effect-1}, one obtains the time dependence $\cO(t^{2(K+1)}/n^{K+1})$. Hence, to achieve a target accuracy $\eta>0$, it suffices to choose $n=\cO(t^2\eta^{-1/(K+1)})$. Thus, increasing $K$ improves only the accuracy dependence (from $\eta^{-1}$ to $\eta^{-1/(K+1)}$), while the $t$-dependence remains quadratic.
			\end{rmk}
			
			This shows that the convergence of the quantum Zeno effect can be improved to order $\frac{1}{n^{K+1}}$ by post-processing the probability distributions of $K+1$ parallel Zeno sequences with step counts $n_\ell=m_\ell pn$. Hence, the accuracy scaling improves exponentially in $K$, while the overhead remains moderate: the number of parallel sequences grows only linearly in $K+1$, and for well-conditioned grids the coefficient sum satisfies $S_K=\sum_{\ell=1}^{K+1}|c_\ell|=\cO(\log(K+1))$.

	\subsection{Applications}\label{subsec:applications}
        Finally, we present two brief examples demonstrating the scheme’s broad applicability, which is formulated with sufficient generality to encompass most instances of the Zeno effect.
        This means that every convergence rate of the Zeno effect in the above form \eqref{eq:generalized-Zeno-effect}, even in infinite-dimensional systems, can be improved using the presented multi-product scheme.
		
		\subsubsection{Dynamical decoupling}
		The so-called `Bang-Bang' method is closely related to the quantum Zeno effect (see \cite{Facchi.2004,Burgarth.2022,Burgath.2024}). In a bipartite system with one part being the bath, the idea is to apply frequent strong `kicks' at short intervals to reduce the influence of the bath on the system, thereby decoupling the system from the bath \cite{Viola.1998}. The quantum operation acting on the bipartite system is $\mathbf{1}_{\cH_1}\otimes M$, where $M$ is an ergodic quantum channel on system~$2$. By \cite[Thm.~16]{Burgath.2024}, the peripheral eigenvalues of $M$ are $\lambda_j = e^{2\pi i j/J}$ for $j = 0,\ldots,J-1$, with spectral gap $\delta$ equal to the spectral radius of $M$ on the complement of its peripheral eigenspace. A concrete example is the qubit amplitude-damping channel $M_\gamma(\rho)=K_0\rho K_0^\dagger+K_1\rho K_1^\dagger$ with $K_0=\ketbra{0}{0}+\sqrt{1-\gamma}\ketbra{1}{1}$ and $K_1=\sqrt{\gamma}\ketbra{0}{1}$ for $0<\gamma<1$: it is ergodic with unique fixed point $\ketbra{0}{0}$, has $J=1$ (only peripheral eigenvalue $\lambda_1=1$), and $\delta=\sqrt{1-\gamma}$.
		
		More formally, the quantum dynamical semigroup is defined on a bipartite system of two finite-dimensional Hilbert spaces $\cH_1 \otimes \cH_2$ and a Hamiltonian, which can always be represented by a set of self-adjoint operators $H_1,H_{1,j}\in\cB(\cH_1)$ and $H_2,H_{2,j}\in\cB(\cH_2)$ for all $j\in\{1,...,d\}$ so that
			\begin{equation*}
				H = H_1 \otimes \1 + \1 \otimes H_2 + \sum_{j=1}^d H_{1,j} \otimes H_{2,j}\,.
			\end{equation*}
			The goal is to interact with the dynamics in such a way that the sum over the interaction terms with the bath is transformed into a weighted sum over $H_{1,j}$ in the effective dynamics (see \cite[Def.~15]{Burgath.2024}). Here, we focus on the recent results of \cite{Burgath.2024}, which states in Theorem 16 that dynamical decoupling using a quantum channel $\1 \otimes M$ is possible if and only if $M$ is ergodic, i.e.~the channel admits a unique fixed point. By \cite[Thm.~16]{Burgath.2024}, this property is equivalent to the assumption that the eigenvalues given in \Cref{eq:power-convergence} are not degenerated and of the form $\lambda_j=e^{2\pi i j/J}$ for all $j\in\{1,...,J\}$. In particular the phases are rational so that \Cref{thm:generalized-zeno} improves the convergence rate of the decoupling scheme \cite[Thm.~16]{Burgath.2024}, i.e.~for a fixed $K\in\N$, fixed point $\rho^*=M(\rho^*)$, there is a $c_\ell\in\R^{K+1}$ so that with $H_{dec}=H_1+\sum_{j=1}^{d}\tr[H_{2,j}\rho^*]H_{1,j}$
			\begin{equation*}
				\Bigl\|\sum_{\ell=1}^{K+1}c_{\ell}\bigl((\1\otimes M)e^{-\frac{t}{\ell n J}i[H,\cdot]}\bigr)^{\ell n J} - e^{-ti[H_{dec},\cdot]}\otimes\sum_{j=1}^{J}\lambda_jP_j\Bigr\|_{1\rightarrow1}=\cO\Bigl(\frac{1}{n^{K+1}}\Bigr)\,.
			\end{equation*}
			
		\subsubsection{Bosonic cat code}
			Another example is the bosonic cat code \cite{Azouit.2016, Guillaud.2019, Guillaud.2023}. Let $\cH = L^2(\R)$ be the Fock space of a single-mode bosonic system endowed with an orthonormal (Fock-)basis $\{\ket{n}\}_{n\in\N}$. The annihilation and creation operators are then defined by $a\ket{n} = \sqrt{n}\ket{n-1}$ and $\ad\ket{n} = \sqrt{n+1}\ket{n+1}$, satisfying for all $n\in\N$ the so-called canonical commutation relation (CCR):
			\begin{equation*}
				[a, \ad]\ket{n} = \ket{n}\,,\quad [a,a]\ket{n} = 0\,.
			\end{equation*}
			The number operator is defined as $N\ket{n} = \ad a\ket{n} = n\ket{n}$, and coherent states are the eigenvectors of $a$, denoted by $\ket{\alpha}$, such that $a\ket{\alpha} = \alpha\ket{\alpha}$ for any $\alpha \in \C$. The code space of the bosonic cat code governed by the two-photon driven dissipative process for a given $\alpha \in \R$ is $\cC_2 \coloneqq \spa\left\{\ketbra{\alpha_1}{\alpha_2} : \alpha_1, \alpha_2 \in \{-\alpha, \alpha\}\right\}$. Note that larger code-spaces are desirable due to their error-correcting properties, but they are not relevant to demonstrate the multi-product Zeno effect. The Schrödinger cat states, which represent the logical qubits, are defined by
			\begin{equation*}
				\begin{aligned}
					\ket{CAT_{\alpha}^+} &\coloneqq \frac{\ket{\alpha} + \ket{-\alpha}}{\sqrt{2(1 + e^{-2|\alpha|^2})}}\,,\\
					\ket{CAT_{\alpha}^-} &\coloneqq \frac{\ket{\alpha} - \ket{-\alpha}}{\sqrt{2(1 - e^{-2|\alpha|^2})}}\,.
				\end{aligned}
			\end{equation*}
			Then, the $X(\theta)$ gate, described in \cite{Mirrahimi.2014} and experimentally realized in \cite{Touzard.2018}, is defined by
			\begin{equation*}
				X(\theta) = \cos\left(\frac{\theta}{2}\right)(P^+_{\alpha} + P^-_{\alpha}) + i\sin\left(\frac{\theta}{2}\right)X_{\alpha}\,,
			\end{equation*}
			with $P^{\pm}_{\alpha} = \ket{CAT^{\pm}_{\alpha}}\bra{CAT^{\pm}_{\alpha}}$ and $X_{\alpha} = \ket{CAT^+_{\alpha}}\bra{CAT^-_{\alpha}} + \ket{CAT^-_{\alpha}}\bra{CAT^+_{\alpha}}$. One possible construction is obtained by intersecting the driving Hamiltonian $H = a + \ad$ with the projections $P_{\alpha} = P^+_{\alpha} + P^-_{\alpha}$ (see \cite{Moebus.2024}), because the Zeno dynamics is
			\begin{equation*}
				e^{itP_{\alpha} HP_{\alpha}}P_{\alpha} = \cos(t(\alpha + \alpha^\dagger))P_{\alpha} + i\sin(t(\alpha + \alpha^\dagger))X_{\alpha}\,,
			\end{equation*}
			which constructs the rotation $X(\theta)$ for appropriate $t$. Since $P_{\alpha}H^k$ and $H^kP_{\alpha}$ are bounded for any $k \in \N$ (as $P_{\alpha}$ is a projection onto a finite-dimensional subspace), the quantum operation $\cP_\alpha = P_\alpha\,\cdot\,P_\alpha$ has a single peripheral eigenvalue $\lambda_1 = 1$ (i.e.\ $J = 1$) and $\delta = 1$, with all non-peripheral eigenvalues equal to zero. \Cref{thm:multi-product-zeno-effect-1} improves the convergence rate to any power $K \in \N$, i.e., for $\cP_{\alpha} = P_{\alpha} \cdot P_{\alpha}$, we have
			\begin{equation*}
				\Bigl\|\sum_{\ell=1}^{K+1}c_\ell \left(\cP_{\alpha} e^{-i\frac{t}{\ell n}[H, \cdot]}\right)^{\ell n} - e^{-ti\cP_{\alpha}[H, \cdot]\cP_{\alpha}}\cP_{\alpha} \Bigr\|_{1 \rightarrow 1} = \cO\Bigl(\frac{1}{n^{K+1}}\Bigr)\,.
			\end{equation*}
			\begin{rmk}[Multi-mode systems]
				The analysis applies only to single or few-mode systems. Extending to multi-mode dynamics encounters fundamental obstructions related to information propagation, resulting in bad dependence on the number of modes: Lieb--Robinson bounds (LRBs), which control information spread, exist only weakly for the Bose--Hubbard model \cite{Eisert.2009,Kuwahara.2021}. Establishing strong locality tools for dissipative bosonic systems (see \cite{Moebus.2023Learning, Rico.2025}) is an important open problem for scalable bosonic quantum error correction.
			\end{rmk}

\section{Conclusion}\label{sec:conclusion}
	This work proposes a multi-product scheme for the quantum Zeno effect by deriving a higher-order error expansion. The approach involves calculating higher-order error terms in the Zeno sequence using a modified Chernoff Lemma, a modified Dunford-Segal approximation, and the holomorphic functional calculus. A multi-product scheme is then applied to improve the convergence rate from order $\frac{1}{n}$ to $\frac{1}{n^{K+1}}$ for a given $K\in\N$. Due to the generality of the result, we are able to apply the result to the `Bang-Bang' method, which is used to decouple systems from their environment, and the bosonic cat code to engineer the $X(\theta)$ gate. It is important to highlight that the scheme is broadly applicable due to its general formulation. However, several research directions remain open. For instance, obtaining tighter bounds, especially for the generalized Zeno effect, would be of interest. Additionally, one could investigate generalizing the choice $n_\ell = \ell n$ to minimize error for a given gate complexity (or, more generally, resource constraints). In this case, an analysis of the Vandermonde matrix discussed in \Cref{LGS}, similar to that in \cite{Low.2019, Carrera.2023}, would be necessary to ensure a numerically well-conditioned problem. Another path could involve extending the approach to time-dependent systems \cite{Watkins.2024,Mizuta.2024,Cao.2025} or relaxing the assumptions on the unbounded generator.

\vspace{2ex}
\emph{Acknowledgments:} 
I would like to thank Paul Gondolf, Robert Salzmann, and the anonymous Reviewer for their valuable discussions on the topic, and Cambyse Rouzé for his advice and feedback. This project was funded within the QuantERA II Programme that has received funding from the EU’s H2020 research and innovation programme under the GA No 101017733. Moreover, I acknowledges the support of the Deutsche Forschungsgemeinschaft (DFG, German Research Foundation) – Project-ID 470903074 – TRR 352 and the support from the Open Access Publishing Fund of the University of T\"{u}bingen.

\sloppy
\setlength{\bibitemsep}{0.05ex}
\printbibliography[heading=bibnumbered]

\newpage

\appendix
\addtocontents{toc}{\protect\setcounter{tocdepth}{1}}
\addcontentsline{toc}{section}{Appendices} 
\addtocontents{toc}{\protect\setcounter{tocdepth}{0}}
\section{Proofs of \Cref{subsec:projective-zeno} (\nameref*{subsec:projective-zeno})}\label{appx-sec:projective-zeno}
	In this section, we state all proofs required in \Cref{subsec:projective-zeno} in detail.
	\subsection{Proofs of auxiliary results}\label{appx-sec:projective-zeno-auxiliary}
		\begin{proof}[Proof of \Cref{lem:chernoff}]
			We follow the proof in \cite[Lem.~4.2]{Moebus.2022}. First, note that $(e^{t(C - \1)})_{t \geq 0}$ defines a contraction semigroup, and $C_\tau \coloneqq (1-\tau)\1 + \tau C = \1 + \tau(C - \1)$ is a contraction for all $\tau \in [0,1]$. By the fundamental theorem of calculus, we have:
			\begin{align*}
				\left(C^n - e^{n(C - \1)}\right) &= n \int_{0}^{1} C_\tau^{n-1} \bigl((C - \1) - C_\tau(C - \1)\bigr) e^{(1-\tau) n(C - \1)} \, d\tau\\
				&= -n(C - \1)^2 \int_{0}^{1} \tau C_\tau^{n-1} e^{(1-\tau) n(C - \1)} \, d\tau\,.
			\end{align*}
			Under the assumption that $C: [0,1] \rightarrow \cB(\cX),\, s \mapsto C(s)$ is a continuously differentiable operator-valued map with $C(0) = \1$, the fundamental theorem of calculus gives:
			\begin{align*}
				\Bigl(C^n(s) - e^{n(C(s) - \1)}\Bigr) &= -ns^2 \iiint_{0}^{1} \tau_1 e^{(1-\tau_1) n(C(s) - \1)} C'(s\tau_2) C'(s\tau_3) C_{\tau_1}^{n-1}(s) \, d\tau_{321}\,,
			\end{align*}
			which completes the proof.
		\end{proof}
		
		\begin{proof}[Proof of \Cref{lem:dunford-segal}]
			To prove this statement, we use the same strategy as in \cite{Moebus.2019, Moebus.2022} and apply the fundamental theorem of calculus three times:
			\begingroup\small
			\begin{equation*}
				\begin{aligned}
					\Bigl(e^{\frac{t}{s_1}(C(s_2) - \1)} - e^{tC'}\Bigr) &= t \int_0^1 e^{(1-\tau_1) \frac{t}{s_1}(C(s_2) - \1)} \biggl(\frac{1}{s_1}(C(s_2) - \1) - C'\biggr) e^{\tau_1tC'} \, d\tau_1 \\
					&= t \iint_0^1 e^{(1 - \tau_1) \frac{t}{s_1}(C(s_2) - \1)} \biggl(\frac{s_2}{s_1} \Bigl(C'(\tau_2 s_2) - C'\Bigr) + \Bigl(\frac{s_2}{s_1} - 1\Bigr) C' \biggr) e^{\tau_1tC'} \, d\tau_{21}\\
					&= t \iiint_0^1 e^{(1 - \tau_1) \frac{t}{s_1}(C(s_2) - \1)} \biggl(\frac{s_2^2}{s_1} \tau_2 C''(\tau_3\tau_2 s_2) + \frac{s_2-s_1}{s_1} C' \biggr) e^{\tau_1tC'} \, d\tau_{321}\\
					&= t \frac{s_2^2}{s_1} \iiint_0^1 e^{(1 - \tau_1)\frac{t}{s_1}(C(s_2) - \1)} \biggl(\tau_2 C''(\tau_3\tau_2 s_2) + \frac{s_2-s_1}{s_2^2} C' \biggr) e^{\tau_1tC'} \, d\tau_{321}.
				\end{aligned}
			\end{equation*}
			\endgroup
			Next, it is clear that $(e^{tC'})_{t\geq0}$ is a semigroup, and applying the above identity shows:
			\begin{equation*}
				\begin{aligned}
					\|e^{tC'}\|_{\cX \rightarrow \cX} &\leq \frac{ts}{2} \max_{\tau \in [0,1]} \|C''(\tau)\|_{\cX \rightarrow \cX} e^{t \|C'\|_{\cX \rightarrow \cX}} + \|e^{\frac{t}{s}(C(s) - \1)}\|_{\cX \rightarrow \cX}\\
					&\leq s t^2 \max_{\tau \in [0,1]} \|C''(\tau)\|_{\cX \rightarrow \cX} e^{t \|C'\|_{\cX \rightarrow \cX}} + 1
				\end{aligned}
			\end{equation*}
			for all $s \in (0,1]$, so that $\|e^{tC'}\|_{\cX \rightarrow \cX} \leq 1$, which completes the proof.
		\end{proof}

	\subsection{Proof of the projective Zeno higher order expansion}\label{appx-sec:projective-zeno-main}
		\begin{proof}[Proof of \Cref{prop:projective-zeno-second-order}]
			First, we absorb the time $t \geq 0$ into $\cL$ and begin with the well-known first-order approximation (see \cite{Moebus.2019, Moebus.2022}). To this end, we define $C(s) \coloneqq P e^{s \cL} P$ on $\cB(P \cX)$ with continuous derivatives $C'(s) = P \cL e^{s \cL} P = P e^{s \cL} \cL P$, $C''(s) = P \cL e^{s \cL} \cL P$ in the uniform operator topology, and $C'''(s) = P \cL^2 e^{s \cL} \cL P$ continuous in the strong operator topology. Since $P \cL P$ is bounded, $(e^{t C'})_{t \geq 0}$ is a well-defined contraction semigroup (see \Cref{lem:dunford-segal}). Then, for $n, m \in \N$,
			\begin{equation}\label{eq:first-order-zeno}
				\begin{aligned}
					C^m(&\tfrac{1}{n}) - e^{C'} \\
					&= \underbrace{C^m(\tfrac{1}{n}) - e^{m(C(\frac{1}{n})-\1)}}_{\Cref{lem:chernoff}} + \underbrace{e^{m(C(\frac{1}{n})-\1)} - e^{C'}}_{\Cref{lem:dunford-segal}} \\
					&= -\frac{m}{n^2} \iiint_0^1 \tau_1 e^{(1-\tau_1) m(C(\frac{1}{n}) - \1)} C'(\tfrac{\tau_2}{n}) C'(\tfrac{\tau_3}{n}) C_{\tau_1}^{m-1}(\tfrac{1}{n}) \, d\tau_{321} \\
					&\qquad + \frac{m}{n^2} \iiint_0^1 e^{(1-\tau_1) m(C(\frac{1}{n}) - \1)} \biggl( \tau_2 C''(\tfrac{\tau_3 \tau_2}{n}) - (n-m) \frac{n}{m} C' \biggr) e^{\tau_1 C'} \, d\tau_{321}.
				\end{aligned}
			\end{equation}
			Note that the identity in $\cB(P \cX)$ is $P$, but it is denoted as usual by $\1$, and for $m=n$, this implies the usual Zeno bound. In the next step, we consider the second-order expansion of the Zeno product for $x \in P \cX$, which is required due to the continuity of $C''(s)$ in the strong operator topology:
			\begingroup\small
			\begin{align}
				C&^n(\tfrac{1}{n})x - e^{C'}x \nonumber \\
				&= \frac{1}{n} \iiint_0^1 e^{(1-\tau_1) n(C(\frac{1}{n}) - \1)} \Bigl( \tau_2 C''(\tfrac{\tau_3 \tau_2}{n}) e^{\tau_1 C'} - \tau_1 C'(\tfrac{\tau_2}{n}) C'(\tfrac{\tau_3}{n}) C_{\tau_1}^{n-1}(\tfrac{1}{n}) \Bigr) x \, d\tau_{321} \nonumber \\
				\begin{split}
					&= \frac{1}{n^2} \iiint_0^1 \biggl[ (1-\tau_1) \iiint_0^1 \xi_2 e^{(1-\xi_1)(1-\tau_1) n(C(\frac{1}{n}) - \1)} C''(\tfrac{\xi_3 \xi_2}{n}) e^{\xi_1(1-\tau_1) C'} d\xi_{321} \biggr] \\
					&\qquad\qquad\qquad\qquad \cdot \Bigl( \tau_2 C''(\tfrac{\tau_3 \tau_2}{n}) e^{\tau_1 C'} - \tau_1 C'(\tfrac{\tau_2}{n}) C'(\tfrac{\tau_3}{n}) C_{\tau_1}^{n-1}(\tfrac{1}{n}) \Bigr) x \, d\tau_{321}
				\end{split} \label{eq:zeno-proj-expansion-1} \\
				\begin{split}
					&\,+\frac{1}{n^2} \iiint_0^1 e^{(1-\tau_1) C'} \Bigl( \tau_2 \biggl[ \tau_3 \tau_2 \int_0^1 C'''(\tfrac{\xi_1 \tau_3 \tau_2}{n}) d\xi_1 \biggr] e^{\tau_1 C'} \\
					&\qquad\qquad\qquad\qquad - \tau_1 \biggl[ \int_0^1 \tau_2 C''(\xi_1 \tfrac{\tau_2}{n}) C'(\xi_2 \tfrac{\tau_3}{n}) + \tau_3 C'(\xi_1 \tfrac{\tau_2}{n}) C''(\xi_1 \tfrac{\tau_3}{n}) d\xi_1 \biggr] C_{\tau_1}^{n-1}(\tfrac{1}{n}) \Bigr) x \, d\tau_{321}
				\end{split} \label{eq:zeno-proj-expansion-2} \\
				\begin{split}
					&\,-\frac{n-1}{n^3} \int_0^1 e^{(1-\tau_1) C'} \tau_1 (C')^2 \biggl[ \iiint_0^1 e^{(1-\xi_1) (n-1)(C_{\tau_1}(\frac{1}{n}) - \1)} \tau_1 \Bigl( \bigl( \xi_2 C''(\tfrac{\xi_3 \xi_2}{n}) - \frac{n}{n-1} C' \bigr) e^{\xi_1 \tau_1 C'} \\
					&\qquad\qquad\qquad\qquad\qquad\qquad\qquad\qquad\qquad - \xi_1 \tau_1 C'(\tfrac{\xi_2}{n}) C'(\tfrac{\xi_3}{n}) C_{\xi_1 \tau_1}^{n-2}(\tfrac{1}{n}) \Bigr) \, d\xi_{321} \biggr] x \, d\tau_1
				\end{split} \label{eq:zeno-proj-expansion-3} \\
				&\,+\frac{1}{2n} \int_0^1 e^{(1-\tau_1) C'} \Bigl( C'' - (C')^2 \Bigr) e^{\tau_1 C'} x \, d\tau_1\,,\nonumber
			\end{align}
			\endgroup
			where we have iteratively applied previously proven bounds, that is, we approximate the following:
			\begin{itemize}
				\item[Eq.~(\texorpdfstring{\ref{eq:zeno-proj-expansion-1}}{???})] $e^{(1-\tau_1) n(C(\frac{1}{n}) - \1)}$ by $e^{(1-\tau_1) C'}$ using \Cref{lem:dunford-segal} with $t=(1-\tau_1)$ and $s_1=s_2=\frac{1}{n}$;
				\item[Eq.~(\texorpdfstring{\ref{eq:zeno-proj-expansion-2}}{???})] $C''(\tfrac{\tau_3 \tau_2}{n})$ by $C''$, and $C'(\tfrac{\tau_2}{n})$ and $C'(\tfrac{\tau_3}{n})$ by $C'$ using the fundamental theorem of calculus and \cite[Lem.~B.15]{Engel.2000}, ensuring the continuity of the integrand;
				\item[Eq.~(\texorpdfstring{\ref{eq:zeno-proj-expansion-3}}{???})] $C_{\tau_1}^{n-1}(\tfrac{1}{n})$ by $e^{\tau_1C'}$ using \Cref{eq:first-order-zeno} with $m=n-1$ and $C_{\tau_1}(s)=(1-\tau_1)\1+\tau_1 C(s)$ satisfying $C'_{\tau_1}=\tau_1C'$, $C''_{\tau_1}=\tau_1C''$, and
				\begin{equation}\label{eq:concatenate-convex-comb}
					(C_{\tau_1})_{\xi_1}=(1-\xi_1)\1+\xi_1((1-\tau_1)+\tau_1 C)=(1-\xi_1\tau_1)\1+\xi_1\tau_1C\,.
				\end{equation}
			\end{itemize}
			Since $C(s)\coloneqq Pe^{s\cL}P$ has bounded derivatives, we have:
			\begin{equation*}
				\begin{aligned}
					\|C(s)\|_{\cX\rightarrow\cX} &\leq 1\,, \qquad & \|C'(s)\|_{\cX\rightarrow\cX} &\leq \|\cL P\|_{\cX\rightarrow\cX}\,, \\
					\|C''(s)\|_{\cX\rightarrow\cX} &\leq \|P\cL^2\|_{\cX\rightarrow\cX}\,, \qquad & \|C'''(s)\|_{\cX\rightarrow\cX} &\leq \|P\cL^2\|_{\cX\rightarrow\cX} \|\cL P\|_{\cX\rightarrow\cX}
				\end{aligned}
			\end{equation*}
			so that
			\begin{equation*}
				\begin{aligned}
					\|&C^n(\tfrac{1}{n}) - e^{C'} - \frac{1}{2n}\int_0^1 e^{\tau_1C'}(C'' - (C')^2)e^{(1-\tau_1)C'} d\tau_1 \|_{\cX\rightarrow\cX} \\
					&\leq \frac{1}{n^2} \biggl( \frac{1}{8} \|P\cL^2\|^2_{\cX\rightarrow\cX} + \frac{1}{12} \|P\cL^2\|_{\cX\rightarrow\cX} \|\cL P\|^2_{\cX\rightarrow\cX} \\
					&\qquad\qquad + \frac{1}{6} \|P\cL^2\|_{\cX\rightarrow\cX} \|\cL P\|_{\cX\rightarrow\cX} + \frac{1}{2} \|P\cL^2\|_{\cX\rightarrow\cX} \|\cL P\|_{\cX\rightarrow\cX} \\
					&\qquad\qquad + \frac{n-1}{n} \|\cL P\|^2_{\cX\rightarrow\cX} \Bigl( \frac{1}{4} \|P\cL^2\|_{\cX\rightarrow\cX} + \frac{n}{2(n-1)} \|\cL P\|_{\cX\rightarrow \cX} + \frac{1}{6} \|\cL P\|_{\cX\rightarrow\cX}^2 \Bigr) \biggr) \\
					&= \frac{1}{24n^2} \bigl( 3 \|P\cL^2\|^2_{\cX\rightarrow\cX} + 8 \|P\cL^2\|_{\cX\rightarrow\cX} \|\cL P\|_{\cX\rightarrow\cX} (2 + \|\cL P\|_{\cX\rightarrow\cX}) \\
					&\qquad\qquad + 4 \|\cL P\|^3_{\cX\rightarrow\cX} (3 + \|\cL P\|_{\cX\rightarrow\cX}) \bigr) \\
					&\leq \frac{1}{6n^2} \Bigl( \|P\cL^2\|_{\cX\rightarrow\cX} + \|\cL P\|_{\cX\rightarrow\cX} (4 + \|\cL P\|_{\cX\rightarrow\cX}) \Bigr)^2\,,
				\end{aligned}
			\end{equation*}
			which concludes the first step of the proof. Next, the fundamental theorem shows
			\begin{equation}\label{eq:expansion}
				\begin{aligned}
					(P&e^{\frac{1}{n}\cL})^n - e^{P\cL P}P\\ 
					&= (Pe^{\frac{1}{n}\cL}P)^{n-1}Pe^{\frac{1}{n}\cL}(\1-P) + (Pe^{\frac{1}{n}\cL}P)^n - e^{P\cL P}P \\
					&= \frac{1}{n^2}\iint_{0}^{1} \tau_1 C^{n-1}(\tfrac{1}{n})P\cL^2 e^{\frac{\tau_1\tau_2}{n}\cL}(\1-P) d\tau_{21} + C^{n-1}(\tfrac{1}{n})P\cL(\1-P) + C^n(\tfrac{1}{n}) - e^{C'}P\,.
				\end{aligned}
			\end{equation}
			Finally, we apply the result from \Cref{eq:first-order-zeno} and Equations (\ref{eq:zeno-proj-expansion-1}--\ref{eq:zeno-proj-expansion-3}) to the last three terms in the equation above to show
			\begin{equation*}
				\begin{aligned}
					\|&(Pe^{\frac{1}{n}\cL})^n - e^{P\cL P}P - \frac{1}{2n} \int_0^1 e^{\tau_1 P\cL P} P\cL[\cL, P] P e^{(1-\tau_1) P\cL P} d\tau_1 - \frac{1}{n} e^{P\cL P} P\cL (\1-P)\|_{\cX\rightarrow\cX} \\
					&\leq \frac{1}{6n^2} \Bigl( \|P\cL^2\|_{\cX\rightarrow\cX} + \|\cL P\|_{\cX\rightarrow\cX} (4 + \|\cL P\|_{\cX\rightarrow\cX}) \Bigr)^2 \\
					&\qquad + \frac{1}{2n^2} \Bigl( 3 \|P\cL^2\|_{\cX\rightarrow\cX} + \|\cL P\|_{\cX\rightarrow\cX} (2 + \|\cL P\|_{\cX\rightarrow\cX}) \Bigr)\\
					&\leq\frac{1}{2n^2} \Bigl(1+2 \|P\cL^2\|_{\cX\rightarrow\cX} + \|\cL P\|_{\cX\rightarrow\cX} (2 + \|\cL P\|_{\cX\rightarrow\cX}) \Bigr)^2\,,
				\end{aligned}
			\end{equation*}
			which concludes the proof by substituting again $\cL$ with $t\cL$ in the statement.
		\end{proof}

		\begin{proof}[Proof of \Cref{thm:projective-zeno-higher-order}]
			We first consider the symmetric case and prove 
			\begin{equation*}
				\Bigl\|(Pe^{\frac{t}{n}\cL}P)^n - e^{tP\cL P}P - \sum_{k=1}^{K} \frac{1}{n^k} \widetilde{E}_k\|_{1 \rightarrow 1} = \cO\Bigl(\frac{1}{n^{K+1}}\Bigr)\, 
			\end{equation*}
			for a sequence $(\widetilde{E}_k)_{k=1}^K\subset\cB(\cX)$. In particular, we prove by induction that for every $K$, there exists a sequence of bounded operators $(\widetilde{E}_k)_{k=1}^K$ such that the following equality holds:
			\begingroup\small
			\begin{equation}\label{eq:proof-expansion-K}
				(Pe^{\frac{t}{n}\cL}P)^n - e^{tP\cL P}P - \sum_{k=1}^{K} \frac{1}{n^k} \widetilde{E}_k = \frac{1}{n^{K+1}} \int_{[0,1]^{3^K}} f_{A_1} \bigl( X_n^{A_2}(\tau^{A_3}), Y_n^{A_2}(\tau^{A_3},\tau^{A_4}), Z_n^{A_2}(\tau^{A_3}) \bigr) d\tau\,.
			\end{equation}
			\endgroup
			Here, $f$ represents exactly what we saw in Equations (\ref{eq:zeno-proj-expansion-1}-\ref{eq:zeno-proj-expansion-3}), namely a sum of products of five operators that appear repeatedly. These are defined by
			\begin{equation}\label{eq:expansion-sym-projective-zeno}
				\begin{aligned}
					X_n^k(s) &\coloneqq C^{(k+1)}\Bigl(\frac{s}{n}\Bigr)\qquad &&\text{and}& \qquad X_{\infty}^k &= C^{(k+1)}(0)\,,\\
					Y_n^k(s_1,s_2) &\coloneqq e^{s_1(n-k)(C_{s_2}(\frac{1}{n})-\1)}\qquad &&\text{and}& \qquad Y_{\infty}(s) &= e^{sC'}\,,\\
					Z_n^k(s) &\coloneqq \bigl(C_{s}(\tfrac{1}{n})\bigr)^{n-k}\,,
				\end{aligned}
			\end{equation}
			for $s,s_1,s_2\in[0,1]$ and $k\in\{0,\dots,K\}$. As can be observed in the proof of \Cref{prop:projective-zeno-second-order}, the input $s$ is a product of different integration variables $\tau_1,\tau_2,\dots$. In addition to these parameters, the order of the operators is important since they do not commute. To account for this, we use the following triple of matrices $(A_1, A_2, A_3,A_4)$ with 
			\begin{itemize}
				\item $A_1 \in \{X, X_{\infty}, Y, Y_{\infty}, Z\}^{3^K \times 4K}$, which determines the order of the products by its rows; for example, the matrix row $(X, Z, Y)$ refers to the operator product $XYZ$, with the operators specified as follows.
				\item $A_2 \in \{0, \dots, K\}^{3^K \times 4K}$, which refers to the different $k$-inputs for the operator-valued functions $X_n^k(s)$, $Y_n^k(s)$, and $Z_n^k(s)$. Here, we assume that the sum of each row, if the entry refers to $X$, $X_\infty$, $Y$ or $Y_\infty$, is less than or equal to $2K+1$.
				\item $A_3,A_4 \in \{0, 1\}^{3^K \times 4K \times 3^K}$, which refers to the different $s$-inputs. As mentioned, these are products of the integration variables, so that $\tau^{A_3}(i,j) = \prod_{k=1}^{3^K} \tau_k^{A_2(i,j,k)}$. These products, organized in the matrix $3^K \times 4K$, are then used as inputs again ($A_4$ resp.). 
			\end{itemize}
			In summary, the function $f$ is a sum over $3^K$ products of $4K$ different operators and coefficients given by monomials in $\tau_1,\dots,\tau_{3^K},\tfrac{1}{n}$, with coefficients bounded in absolute value by $K-1$. The matrices $X_n^{A_2}(s^{A_3})$, $Y_n^{A_2}(s^{A_3},s^{A_4})$, and $Z_n^{A_2}(s^{A_3})$ are elements of $\cB(\cX)^{4K \times 3^K}$ with operator-valued entries. These determine the products in the function $f$ via the matrix $A_1$, which includes the information on whether $X, X_{\infty}, Y, Y_{\infty}, Z$ is used at a position $i$ in the product in the $j^{\text{th}}$ term, and with which parameters $A_2(i,j)$, $A_3(i,j)$, and $A_4(i,j)$ as input.
			
			It is clear that for $K=1$, the above construction corresponds to the expression considered in Equations (\ref{eq:zeno-proj-expansion-1}-\ref{eq:zeno-proj-expansion-3}) in the proof of \Cref{prop:projective-zeno-second-order}, verifying the `induction start'.
			
			Next, we assume that the expansion in \Cref{eq:proof-expansion-K} is given for some $K \in \N$. We then repeat the following expansions of $X_n^k(s)$, $Y_n^k(s_1,s_2)$, and $Z_n^k(s)$:
			\begingroup\footnotesize
			\begin{equation*}
				\begin{aligned}
					X_n^k(s)&=C^{(k)}(0)+\frac{s}{n}\int_{0}^{1}C^{(k+2)}\left(\frac{\tau_1 s}{n}\right)d\tau_{1}\\
					&=X_{\infty}^{k}+\frac{s}{n}\int_{0}^{1}X_n^{k+1}(\tau_1 s)d\tau_{1}\\
					Y_n^k(s_1,s_2)&=e^{s_1s_2C'} + s_1s_2 \frac{n-k}{n^2} \iiint_0^1 e^{(1 - \tau_1)s_1(n-k)(C_{s_2}(\frac{1}{n}) - \1)} \biggl(\tau_2 C''(\tfrac{\tau_3\tau_2}{n})\\
					&\qquad\qquad\qquad\qquad\qquad\qquad\qquad\qquad\qquad\qquad\qquad\qquad -k\frac{n}{n-k} C' \biggr) e^{\tau_1s_1s_2C'} \, d\tau_{321}\\
					&=Y_{\infty}(s_1s_2) + s_1s_2 \frac{n-k}{n^2} \iiint_0^1 Y_n^k((1 - \tau_1) s_1,s_2)\biggl(\tau_2 X_{n}^{1}(\tau_3\tau_2)\\
					&\qquad\qquad\qquad\qquad\qquad\qquad\qquad\qquad\qquad\qquad\qquad\qquad -k\frac{n}{n-k} X^{0}_\infty \biggr) Y_\infty(\tau_1s_1s_2)\, d\tau_{321}\\
					Z_n^k(s)&= e^{sC'} -s^2\frac{n-k}{n^2} \iiint_0^1 \tau_1 e^{(1-\tau_1) (n-k)(C_s(\frac{1}{n}) - \1)} C'(\tfrac{\tau_2}{n}) C'(\tfrac{\tau_3}{n}) C_{\tau_1s}^{n-k-1}(\tfrac{1}{n}) \, d\tau_{321} \\
					&\qquad\qquad + s\frac{n-k}{n^2} \iiint_0^1 e^{(1-\tau_1) (n-k)(C_s(\frac{1}{n}) - \1)} \biggl( \tau_2 C''(\tfrac{\tau_3 \tau_2}{n}) - k\frac{n}{n-k} C' \biggr) e^{\tau_1s C'} \, d\tau_{321}\\
					&= Y_{\infty}(s) -s^2\frac{n-k}{n^2} \iiint_0^1 \tau_1 Y_n^k((1-\tau_1),s)X_n^{0}(\tau_2)X_n^{0}(\tau_3)Z^{k+1}_n(\tau_1s)\, d\tau_{321} \\
					&\qquad\qquad + s\frac{n-k}{n^2} \iiint_0^1 Y_n^k((1-\tau_1),s) \biggl( \tau_2 X^1_n(\tau_3 \tau_2) - k\frac{n}{n-k} X^1_{\infty} \biggr) Y_{\infty}(\tau_1s) \, d\tau_{321}
				\end{aligned}
			\end{equation*}
			\endgroup
			using the fundamental theorem of calculus, \Cref{lem:dunford-segal}, and \Cref{eq:first-order-zeno} with \Cref{eq:concatenate-convex-comb}. Since the integration variables are extended by three dimensions, the number of terms increases by a factor at most three, the products increase by four, and the $k$'s increase by at most two. We achieve a $1/n$ error increase by one power if and only if there is an $n$-dependent operator left in the product; if not, the term is added to the definition of $\widetilde{E}_{K+1}$.
			
			Moreover, note that the following bounds hold:
			\begin{equation*}
				\begin{aligned}
					\|X_n^k(s)\|_{\cX \rightarrow \cX}\,,\,\, \|X_{\infty}^k\|_{\cX \rightarrow \cX}&\leq \|P\cL^k\|\|\cL P\|_{\cX \rightarrow \cX}\\
					\|Y_n^k(s)\|_{\cX \rightarrow \cX}\,,\,\, \|Y_{\infty}(s)\|_{\cX \rightarrow \cX}\,,\,\, \|Z_n^k(s)\|_{\cX \rightarrow \cX}&\leq1
				\end{aligned}
			\end{equation*}
			and are independent of $n$, which concludes the expansion of the symmetric, projective Zeno effect (\ref{eq:expansion-sym-projective-zeno}). Next, we consider \Cref{eq:expansion}
			\begin{equation*}
				\begin{aligned}
					(P&e^{\frac{1}{n}\cL})^n - e^{tP\cL P}P \\ 
					&= C^{n-1}\left(\frac{1}{n}\right)Pe^{\frac{1}{n}\cL}(\1-P) + \left(C^n\left(\frac{1}{n}\right) - e^{C'}\right)P\,.
				\end{aligned}
			\end{equation*}
			Since $C^{n-1}\left(\frac{1}{n}\right)$ and $\left(C^n\left(\frac{1}{n}\right) - e^{C'}\right)P$ can be expanded analogously to the above proof and
			\begin{equation*}
				\begin{aligned}
					Pe^{\frac{1}{n}\cL}(\1-P)
					&= \sum_{k=0}^{K}\frac{1}{n^k}P\cL^k(\1-P)\\
					&\quad + \frac{1}{n^{K+1}}\int_{[0,1]^{K+1}}\tau_1^{K}\tau_2^{K-1}\cdots\tau_{K+1}
					P\cL^{K+1}e^{\frac{\tau_1\cdots\tau_{K+1}}{n}\cL}(\1-P)d\tau\,.
				\end{aligned}
			\end{equation*}
			by the fundamental theorem of calculus, as long as $P\cL^{K+1}$ is bounded by assumption, the proof is complete.
		\end{proof}

\section{Proofs of \Cref{subsec:generalized-zeno} (\nameref*{subsec:generalized-zeno})}\label{appx-sec:generalized-zeno}
	In this section, we state all proofs required in \Cref{subsec:generalized-zeno} in detail.
	
	\subsection{Proofs of auxiliary result}\label{appx-sec:generalized-zeno-auxiliary}
		\begin{proof}[Proof of \Cref{lem:functional-calculus}]
			In the first part of the proof, we follow the strategy of \cite[Thm.~III]{Moebus.2022}, adapted to the special case where, for all $x \in \cX$,
			\begin{equation*}
				M(t)x \coloneqq Me^{t\cL}x = Mx + t\int_{0}^1M\cL e^{t\tau\cL}x\,d\tau \eqqcolon Mx + tB(t)x\,.
			\end{equation*}
			Since $M\cL$ is bounded by assumption, we have $\|M(t) - M(0)\|_{\cX\rightarrow\cX} \leq t\|M\cL\|_{\cX\rightarrow\cX}$, so \cite[Lem.~A.2]{Moebus.2022} shows the semicontinuity of the spectrum under this perturbation. That is, there exists an $\epsilon > 0$ such that $\spec(Me^{t\cL})$ remains separated by $\Gamma$ for all $t \in [0,\epsilon]$. By applying the holomorphic functional calculus \cite[Thm.~2.3.1-3]{Simon.2015}, we obtain that
			\begin{equation*}
				P(t) = \frac{1}{2\pi i} \oint_{\Gamma}R(z,Me^{t\cL})\,dz
			\end{equation*}
			defines, for all $t \in [0,\epsilon]$, a projection onto the eigenspace enclosed by $\Gamma$. Next, for simplicity, we introduce the notation $R(t) = R(z,M + tB(t))$ and $R = R(0)$. Applying the second resolvent identity, $tR(t)B(t)R = R - R(t)$, and using the fundamental theorem of calculus, we show that
			\begin{equation*}
				\begin{aligned}
					2\pi i P(t) &= \oint_{\Gamma}R(t)\,dz \\
					&= \oint_{\Gamma}\Bigl(R - tR(t)B(t)R\Bigr)\,dz \\
					&= \int_0^1 \oint_{\Gamma}\Bigl(R - tRBR + t^2R(t)(BR)^2 - t^2 RB'(\tau_1 t)R\Bigr)\,dz\,d\tau_1 \\
					&= \oint_{\Gamma}\Bigl(R - tRBR + t^2\bigl(R(BR)^2 - RB'R\bigr)\Bigr)\,dz \\
					&\quad - t^3 \iint_0^1 \oint_{\Gamma}R(t)(BR)^3\,dz\,d\tau_{12}\\
					&\quad + t^3 \iint_0^1 \oint_{\Gamma}
					R\bigl(B'(\tau_2 t)RB(t) + BRB'(\tau_2 t) - \tau_1 B''(\tau_1 \tau_2 t)\bigr)R\,dz\,d\tau_{12}\,.
				\end{aligned}
			\end{equation*}
			Since $B^{(k)}(t) = \int_0^1 M\cL^{k+1} e^{\tau t \cL}\,d\tau$ and $M\cL^{k+1}$ is bounded for all $k \in \{0,\dots, K\}$, we have
			\begin{equation*}
				\|B^{(k)}(t)\|_{\cX\rightarrow \cX} \leq \|M\cL^{k+1}\|_{\cX\rightarrow \cX}\,,
			\end{equation*}
			and the resolvent $R(z,M + tB(t))$ is uniformly bounded for $z \in \Gamma$ and $t \in [0,\epsilon]$. Iterating the above strategy completes the proof.
		\end{proof}

	\subsection{Proof of the generalized Zeno higher order expansion}\label{appx-sec:generalized-zeno-main}
		\begin{proof}[Proof of \Cref{thm:generalized-zeno}]
			As mentioned above, the uniform power convergence assumption is equivalent to the spectral gap assumption with zero nilpotent operators \cite[Prop.~3.1]{Becker.2021}. This is why there exists a curve $\gamma:[0,1] \rightarrow \C, \gamma(s) = \frac{2\delta+1}{3} e^{s 2\pi i}$ that separates the part of the spectrum of $M$ inside $\D_{\delta}$ from the isolated eigenvalues $\{\lambda_j\}_{j=1}^J$, along with similar curves $\Gamma_1, ..., \Gamma_J$ that separate each isolated eigenvalue (see \Cref{fig1}). Then, \Cref{lem:functional-calculus} shows that these curves separate the spectrum of $Me^{t\cL}$ for sufficiently small $t \in [0, \varepsilon]$. This allows us to perform the following splitting:
			\begin{equation}\label{eq:splitting}
				(Me^{\frac{1}{n}\cL})^n = \sum_{j=1}^{J}\lambda_j^n \left( \overline{\lambda}_j P_j\left(\tfrac{1}{n}\right)Me^{\frac{1}{n}\cL}P_j\left(\tfrac{1}{n}\right) \right)^n + \frac{1}{2\pi i}\oint_{\gamma} z^n R(z,Me^{\frac{1}{n}\cL})\, dz\,,
			\end{equation}
			with the eigenprojection
			\begin{equation*}
				P_j(t) = \frac{1}{2\pi i} \oint_{\Gamma_j}R(z,Me^{t\cL})\,dz
			\end{equation*}
			defined in \Cref{lem:functional-calculus}. The last term in \Cref{eq:splitting} is strictly contractive because
			\begin{equation*}
				\left\|\frac{1}{2\pi i}\oint_{\gamma} z^n R(z,Me^{\frac{1}{n}\cL}) \, dz \right\|_{\cX\rightarrow \cX} \leq c \widetilde{\delta}^n,
			\end{equation*}
			where $c = \max_{z \in \gamma}\|R(z,M)\|_{\cX\rightarrow\cX} \frac{2+2\widetilde{\delta}}{1+2\widetilde{\delta}}$ with $\widetilde{\delta} = \frac{1+2\delta}{3}$, holds true due to the second von Neumann series and the boundedness of the resolvent (see \cite[Eq.~23-24]{Moebus.2022}). Therefore, we consider the following operator on $\cX$ for any $j\in\{1, ..., J\}$:
			\begin{equation*}
				C(s) \coloneqq \overline{\lambda}_j P(s)Me^{s\cL}P(s).
			\end{equation*}
			and omit the $j$ in the notation so that $P(t)=P_j(s)$, $P=P(0)$, and $P'=P'(s)$. It admits the following expansion in the operator norm:
			\begin{equation*}
				\begin{aligned}
					C(s) &= \overline{\lambda} \left( P + \sum_{k=1}^K s^k F_k \right) \left( \sum_{k=0}^{K}\frac{s^k}{k!} M\cL^k \right) \left( P + \sum_{k=1}^K s^k F_k \right) + \cO(s^{K+1})\\
					&= P + \overline{\lambda}(PMP'+P'MP+PM\cL P) + \overline{\lambda} \sum_{k=2}^{K} s^k \sum_{k_1+k_2+k_3=k} F_{k_1} G_{k_2} F_{k_3} + \cO(s^{K+1})\\
					&= P + s(P'+P\cL P) + \overline{\lambda} \sum_{k=2}^{K} s^k \sum_{k_1+k_2+k_3=k} F_{k_1} G_{k_2} F_{k_3} + \cO(s^{K+1})\,,
				\end{aligned}
			\end{equation*}
			for $s \rightarrow 0$, where $k_1, k_2, k_3 \in \{0, ..., K\}$, $F_0 \coloneqq P$, and $G_{k_2} \coloneqq \frac{1}{k!} M\cL^k$. Here, we used the same trick as in \cite{Moebus.2019}, $P'=\frac{d}{dt}P(t)\bigl|_{t=0}=\frac{d}{dt}P^2(t)\bigl|_{t=0}=P'P+PP'$. Next, we apply the modified Chernoff Lemma (\ref{lem:chernoff}) to $C(s)$ on $P(s)\cX$ and then use the fundamental theorem of calculus to rewrite $(C-\1)\equiv C(s)-P(s)$, which implies the following bound on $\cX$:
			\begingroup\footnotesize
			\begin{align*}
				C^n(s) &- e^{n(C(s) - P(s))} P(s) \\
				&= -ns^2 \iiint_{0}^{1} \tau_1 e^{(1-\tau_1) n(C(s) - \1)} (C'(s\tau_2) - P'(s\tau_2)) (C'(s\tau_3) - P'(s\tau_3)) C_{\tau_1}^{n-1}(s) \, d\tau_{321}\,.
			\end{align*}
			\endgroup
			Similarly, the modified Dunford-Segal approximation \eqref{lem:dunford-segal} is adapted to
			\begingroup\footnotesize
			\begin{equation*}
				\begin{aligned}
					\Bigl(&e^{\frac{t}{s_1}(C(s_2) - P(s_2))} - e^{t\widetilde{C}'}\Bigr) \\
					&= t \int_0^1 e^{(1-\tau_1) \frac{t}{s_1}(C(s_2) - P(s_2))} \biggl(\frac{1}{s_1}(C(s_2) - P(s_2)) -\widetilde{C}'\biggr) e^{\tau_1t\widetilde{C}'} \, d\tau_1 \\
					&= t \iint_0^1 e^{(1 - \tau_1) \frac{t}{s_1}(C(s_2) - P(s_2))} \biggl(\frac{s_2}{s_1} \Bigl(C'(\tau_2 s_2)-P'(\tau_2 s_2) - \widetilde{C}'\Bigr) + \Bigl(\frac{s_2}{s_1} - 1\Bigr) \widetilde{C}' \biggr) e^{\tau_1t\widetilde{C}'} \, d\tau_{21}\\
					&= t \iiint_0^1 e^{(1 - \tau_1) \frac{t}{s_1}(C(s_2) - P(s_2))} \biggl(\frac{s_2^2}{s_1} \tau_2 \bigl(C''(\tau_3\tau_2 s_2)-P''(\tau_3\tau_2 s_2)\bigr) + \frac{s_2-s_1}{s_1} \widetilde{C}' \biggr) e^{\tau_1t\widetilde{C}'} \, d\tau_{321}\\
					&= t \frac{s_2^2}{s_1} \iiint_0^1 e^{(1 - \tau_1)\frac{t}{s_1}(C(s_2) - P(s_2))} \biggl(\tau_2\bigl(C''(\tau_3\tau_2 s_2)-P''(\tau_3\tau_2 s_2)\bigr) + \frac{s_2-s_1}{s_2^2} \widetilde{C}' \biggr) e^{\tau_1t\widetilde{C}'}\, d\tau_{321}.
				\end{aligned}
			\end{equation*}
			\endgroup
			with $\widetilde{C}'=C'-P'=P\cL P$. Then, we can repeat the proof of \Cref{thm:projective-zeno-higher-order} with the following set of operators (similar to \Cref{eq:expansion-sym-projective-zeno})
			\begin{equation}\label{eq:expansion-generalized-zeno}
				\begin{aligned}
					W^k_n(s)&\coloneqq P^{(k)}(\tfrac{s}{n})\qquad &&\text{and}& \qquad W_{\infty}^k &= P^{(k)}\,,\\
					X_n^k(s) &\coloneqq C^{(k+1)}(\tfrac{s}{n})\qquad &&\text{and}& \qquad X_{\infty}^k &= C^{(k+1)}(0)\,,\\
					Y_n^k(s_1,s_2) &\coloneqq e^{s_1(n-k)(C_{s_2}(\frac{1}{n})-P(\frac{1}{n}))}\qquad &&\text{and}& \qquad Y_{\infty}(s) &= e^{s(C'-P')}\,,\\
					Z_n^k(s) &\coloneqq \bigl(C_{s}(\tfrac{1}{n})\bigr)^{n-k}\,,
				\end{aligned}
			\end{equation}
			with $C_{\tau}(s)\coloneqq (1-s_2)P(s) + \tau C(s)$. Then, we follow the proof of \Cref{prop:projective-zeno-second-order} with the above operators and use the expansion of $P(s)$, given in \Cref{lem:functional-calculus}. This proves for every $j\in\{1,...,J\}$ the expansion
			\begin{equation*}
				\Bigl\|\left( \overline{\lambda}_j P_j\left(\tfrac{1}{n}\right)Me^{\frac{1}{n}\cL}P_j\left(\tfrac{1}{n}\right) \right)^n - e^{tP_j\cL P_j}P_j - \sum_{k=1}^{K} \frac{1}{n^k} E_{k,j}\Bigr\|_{\cX \rightarrow \cX} = \cO\Bigl(\frac{1}{n^{K+1}}\Bigr)\,,
			\end{equation*}
			which finishes the proof.
		\end{proof}
\end{document}